\documentclass[a4paper,11pt,fleqn]{article}
\usepackage{amsmath,amssymb,amsthm,enumerate}

\allowdisplaybreaks

\newcommand{\p}{\partial}
\newcommand{\ord}{\mathop{\rm ord}\nolimits}
\newcommand{\const}{\mathop{\rm const}\nolimits}
\newcommand{\sgn}{\mathop{\rm sgn}\nolimits}

\newcommand{\lsemioplus}{\mathbin{\mbox{$\lefteqn{\hspace{.70ex}\rule{.4pt}{1.2ex}}{\in}$}}}

\newcommand{\todo}[1][\null]{\ensuremath{\clubsuit}}
\newcommand{\noprint}[1]{}

\newtheorem{theorem}{Theorem}

\newtheorem{corollary}[theorem]{Corollary}
\newtheorem{proposition}[theorem]{Proposition}
{\theoremstyle{definition}

}

\begin{document}

\noindent
{\LARGE\bf Enhanced symmetry analysis\\ of two-dimensional Burgers system\par}

\vspace{4mm}

\noindent
Stavros Kontogiorgis$^\dag$, Roman O. Popovych$^\ddag$ and Christodoulos Sophocleous$^\S$

\vspace{4mm}

\noindent{\it
\hbox to 3 mm{$^{\dag\S}$\hfil}Department of Mathematics and Statistics, University of Cyprus, Nicosia CY 1678, Cyprus\\[1mm]
\hbox to 3 mm{$^\ddag$\hfil}Wolfgang Pauli Institut, Oskar-Morgenstern-Platz 1, 1090 Wien, Austria\\
\hbox to 3 mm{\hfil}Institute of Mathematics of NAS of Ukraine, 3 Tereshchenkivs'ka Str., 01004 Kyiv, Ukraine

}

\vspace{4mm}

{\noindent E-mails:  $^\dag$kontogiorgis.stavros@ucy.ac.cy, $^\ddag$rop@imath.kiev.ua, $^\S$christod@ucy.ac.cy}

\vspace{5mm}\par\noindent\hspace*{5mm}\parbox{150mm}{\small%\looseness=-1
We carry out enhanced symmetry analysis of a two-dimensional Burgers system.
The complete point symmetry group of this system is found using an enhanced version of the algebraic method.
Lie reductions of the Burgers system are comprehensively studied in the optimal way
and new Lie invariant solutions are constructed.
We prove that this system admits no local conservation laws
and then study hidden conservation laws, including potential ones.
Various kinds of hidden symmetries (continuous, discrete and potential ones) are considered for this system as well.
We exhaustively describe the solution subsets of the Burgers system
that are its common solutions with its inviscid counterpart and with the two-dimensional Navier--Stokes equations.
Using the method of differential constraints, which is particularly efficient for the Burgers system,
we construct a number of wide families of solutions of this system
that are expressed in terms of solutions of the (1+1)-dimensional linear heat equation
although they are not related to the well-known linearizable solution subset of the Burgers system.
\par\vspace{5mm}}

\noindent{\footnotesize
\par}

\noprint{MSC:
35-XX   Partial differential equations
 35Bxx  Qualitative properties of solutions
  35B06   Symmetries, invariants, etc.
 35Cxx  Representations of solutions
  35C05   Solutions in closed form
76-XX   Fluid mechanics {For general continuum mechanics, see 74Axx, or other parts of 74-XX}
 76Mxx  Basic methods in fluid mechanics [See also 65-XX]
  76M60   Symmetry analysis, Lie group and algebra methods
}

\section{Introduction}

In the last decades a lot of attention has been paid to study various generalizations of the Burgers equation \cite{sach1987a}.
If we ignore the pressure gradient terms from the momentum equations in the incompressible Navier--Stokes equations,
where the kinematic viscosity is set, without loss of generality, to be equal 1,
we obtain the nonlinear system
\begin{gather}\label{eq:BurgersSystem}
\begin{split}
&u_t+uu_x+vu_y-u_{xx}-u_{yy}=0,\\
&v_t+uv_x+vv_y-v_{xx}-v_{yy}=0,
\end{split}
\end{gather}
which is known as the two-dimensional Burgers system. 
This system and its three-dimensional counterpart were first considered in \cite[Eq.~(61)]{cole1951} 
as a multidimensional analogue of the famous Burgers equation.
We point out that solutions of the system~\eqref{eq:BurgersSystem} do not necessarily satisfy the continuity equation.
The system~\eqref{eq:BurgersSystem} has been considered, for example, in \cite{hlav1983a,sale1987a} and
certain underlying geometric and group theoretical properties were discussed.

It is well known that the Hopf--Cole transformation relates the Burgers equation and the linear heat equation \cite{cole1951,hopf1950}.
Note that this linearization had implicitly been presented earlier in \cite[p.\,102, Exercise~3]{fors1906}.
The Hopf--Cole transformation can be generalized for the above multidimensional analogues of the Burgers equation \cite[Eq.~(63)]{cole1951}, 
see also~\cite[Eq.~(2.22)]{ames1965} and, for a wider context, \cite{broa2015c}.
In space dimension two, a generalization of the Hopf--Cole transformation~is
\begin{gather}\label{eq:2DHopfColeTrans}
u=-2\frac{\phi_x}\phi,\quad v=-2\frac{\phi_y}\phi.
\end{gather}
It reduces the system~\eqref{eq:BurgersSystem} with the additional differential constraint
\begin{gather}\label{eq:BurgersSystemLinearizationConstraint}
u_y=v_x,
\end{gather}
meaning that the flow $(u,v)$ is irrotational,
to the single (1+2)-dimensional linear heat equation
\begin{gather}\label{eq:(1+2)DLinHeatEq}
\phi_t-\phi_{xx}-\phi_{yy}=0.
\end{gather}
In the present paper, Hopf--Cole-type transformations are also derived for certain reduced systems of~\eqref{eq:BurgersSystem}.

Under the other constraint $v=0$,
the second equation of the system~\eqref{eq:BurgersSystem} is satisfied identically,
and its first equation reduces to a (1+2)-dimensional generalization of the Burgers equation,
\begin{gather}\label{(1+2)DBurgersEq}
u_t+uu_x-u_{xx}-u_{yy}=0,
\end{gather}
which was derived in~\cite{raja2008a} as an equation for the wave phase of two-dimensional sound simple waves in weakly dissipative flows.
Therein, symmetry analysis of this equation was carried out,
which included the first exhaustive study of its Lie reductions in an optimized way
and the construction of several new families of its exact solutions.
For the first time, the equation~\eqref{(1+2)DBurgersEq} had formally appeared in~\cite{edwa1995a},
where its maximal Lie invariance algebra had been computed
and its two-step Lie reductions to ODEs had preliminarily been considered.
The equation~\eqref{(1+2)DBurgersEq} was also singled out in~\cite{deme2008a}
in the course of the group classification of (1+2)-dimensional diffusion--convection equations.

Transformation properties of evolution equations and systems have been widely studied
because of many practical benefits that such knowledge gives
and also because of the variety of physical applications that are modeled by these equations.
Particularly useful in the study of a partial differential equation is 
the knowledge of the corresponding Lie (pseudo)group of point symmetries.
While there is no existing general theory for solving nonlinear PDEs,
methods of group analysis of differential equations have been proved to be very powerful.
Lie symmetries of the system~\eqref{eq:BurgersSystem} were considered in several papers \cite{abdu2016a,elsa2014a,tami1991a}.
Although the maximal Lie invariance algebra~$\mathfrak g$ of this system had already been found accurately in~\cite{tami1991a},
the further studies of Lie reductions of the system~\eqref{eq:BurgersSystem} were incomplete or even essentially incorrect.
Optimal lists of subalgebras of the algebra~$\mathfrak g$ were not constructed,
and the used sets of subalgebras had a lot of significant weaknesses.

In the present paper, we carry out enhanced symmetry analysis of the two-dimensional Burgers system~\eqref{eq:BurgersSystem}.
The objects, structures and properties studied in the paper include not only Lie reductions and invariant solutions
but also the complete point symmetry group of the original system~\eqref{eq:BurgersSystem};
the maximal Lie invariance algebras, the complete point symmetry groups and local conservation laws of various reduced systems;
hidden (continuous, discrete and potential) symmetries;
hidden (local and potential) conservation laws;
the solution subsets of the system~\eqref{eq:BurgersSystem} 
that are its common solutions with the inviscid Burgers system and with the two-dimensional Navier--Stokes equations;
solutions affine in the space variables; and
subsets of solutions that are expressed in terms of solutions of the (1+1)-dimensional linear heat equation.

The structure of the maximal Lie invariance algebra of the system~\eqref{eq:BurgersSystem} is discussed
in Section~\ref{sec:BurgersSystemMLIAlgAndCompleteGroup}.
Therein we compute the complete point symmetry group~$G$ of this system and its discrete symmetries
using the algebraic method, which is enhanced with early factoring out the inner automorphisms
related to Levi factors.
In Section~\ref{sec:SymsOfLinearizableSubsetOfSolutions} we prove that
the complete point symmetry group of the joint system of~\eqref{eq:BurgersSystem}
with the differential constraint~\eqref{eq:BurgersSystemLinearizationConstraint},
which singles out the widest linearizable solution subset of~\eqref{eq:BurgersSystem},
coincides with the group~$G$.
The study of Lie reductions of the system~\eqref{eq:BurgersSystem} is based on
constructing optimal lists of one- and two-dimensional subalgebras of its Lie symmetry algebra
in Section~\ref{sec:SubalgebrasOfLieInvarianceAlgebra}.
In order to optimize the construction of Lie submodel hierarchy for the system~\eqref{eq:BurgersSystem},
we exploit a special reduction technique,
which was systematically used in~\cite{fush1994a,fush1994b} and discussed in~\cite{popo1995b}.
This technique is to choose reduction ansatzes among possible ones in such a~way
that reduced systems are of simple and unified forms
being similar to the form of the origin system as much as possible.
Listed inequivalent one-dimensional subalgebras lead to Lie reductions
where each reduced system consists of two PDEs in two independent variables,
see Section~\ref{sec:BurgersSystemLiereductionsOfCodim1}.
In Section~\ref{sec:SymmetryAnalysisOfReducedSystemsOfPDEs}
we study Lie symmetries of all obtained reduced systems of PDEs.
The reduced system for solutions invariant with respect to space shifts
is linearized by a Hopf--Cole-type transformation
to uncoupled system of two copies of the (1+1)-dimensional linear heat equation,
and Lie symmetries of the other reduced system are induced by Lie symmetries
of the original system~\eqref{eq:BurgersSystem}.
This is why we do not need to consider further Lie reductions of reduced systems of PDEs.
Instead of two-step Lie reductions, it is more advantageous
to reduce the system~\eqref{eq:BurgersSystem} directly to systems of ODEs
using the optimal list of two-dimensional subalgebras.
At the same time, essential are only Lie reductions with respect to two-dimensional subalgebras
that, up to $G$-equivalence, contain no vector fields of space shifts.
Inequivalent essential reduced systems of ODEs are collected in Section~\ref{sec:BurgersSystemLieReductionOfCodim2}.
The Burgers system~\eqref{eq:BurgersSystem} admits no useful hidden symmetries related to these reductions.
New exact stationary similarity solutions of~\eqref{eq:BurgersSystem} are constructed
in Section~\ref{sec:SolutionOfReducedSystemsOfODEs}.
In Section~\ref{sec:BurgersSystemCLs}
we prove that the system~\eqref{eq:BurgersSystem} admits no local conservation laws.
Then we discuss local and potential conservation laws of various submodels related to this system.
The solutions of the Burgers system~\eqref{eq:BurgersSystem}
that can be prolonged to solutions of the (1+2)-dimensional Navier--Stokes equations
are comprehensively described in Section~\ref{sec:BurgersSystemSolutionsCommonWithNavier-StokesEqns}.
In Section~\ref{sec:OneMoreReductionToSingle(1+2)DPDE} we study
a new non-Lie reduction of the system~\eqref{eq:BurgersSystem}
to a single (1+2)-dimensional PDE, which is associated with the differential constraint $u_x=v_y$.
We compute the complete point symmetry groups
of the joint system of~\eqref{eq:BurgersSystem} with this constraint
and of the reduced equation.
Looking for exact solutions of the reduced equation results in finding
more stationary similarity solutions of~\eqref{eq:BurgersSystem},
one more family of solutions of~\eqref{eq:BurgersSystem}
that is parameterized by two arbitrary nonvanishing solutions of the (1+1)-dimensional linear heat equation,
and solutions of~\eqref{eq:BurgersSystem} expressed via
the general solution of a simple complex Hamilton--Jacobi equation.
All solutions of the Burgers system~\eqref{eq:BurgersSystem} that are affine in the space variables
are listed in Section~\ref{sec:SolutionsAffineInSpaceVars}.
Section~\ref{sec:CommonSolutionsOfViscidAndInviscidBurgersSystems} is devoted to
the exhaustive description of common solutions of the `viscid' and  `inviscid' (1+2)-dimensional Burgers systems.
In Section~\ref{sec:BurgersSystemMoreSolutionsWithMethodOfDiffConstraints}
we consider the system~\eqref{eq:BurgersSystem} with various differential constraints,
$v=0$, $v_x=0$, $u_{xx}=v_x=0$ and $u_{xx}=v=0$, respectively,
which allows us to construct new wide linearizable subsets of solutions of~\eqref{eq:BurgersSystem}.
Among these subsets, there are families parameterized by one or two arbitrary solutions of the (1+1)-dimensional linear heat equation.
Analyzing results of the present paper in Section~\ref{sec:Conclusion},
we discuss which optimized tools were used to obtain these results.

\section{Lie invariance algebra and complete point symmetry group}\label{sec:BurgersSystemMLIAlgAndCompleteGroup}

The classical approach for deriving Lie symmetries is well known and established in the last decades \cite{BlumanBook1989,olve1993b,ovsiannikov1982}.
The maximal Lie invariance algebra of the Burgers system~\eqref{eq:BurgersSystem}, which was found in~\cite{tami1991a},
is the so-called reduced (i.e., centerless) special Galilei algebra~\cite{FushchychBarannykBarannyk} with space dimension two
\begin{gather*}
\mathfrak g=\langle P^t,D,\Pi,J,P^x,P^y,G^x,G^y\rangle,
\\
\hspace*{-\mathindent}\mbox{where}
\\
P^t=\p_t, \quad
D=2t\p_t+x\p_x+y\p_y-u\p_u-v\p_v,
\\
\Pi=t^2\p_t+tx\p_x+ty\p_y+(x-tu)\p_u+(y-tv)\p_v,\quad
J=x\p_y-y\p_x+u\p_v-v\p_u,
\\
P^x=\p_x, \quad
P^y=\p_y, \quad
G^x=t\p_x+\p_u,\quad
G^y=t\p_y+\p_v.
\end{gather*}
The nonzero commutation relations of~$\mathfrak g$ are
\begin{gather*}
[P^t,D]=2P^t,\quad
[D,\Pi]=2\Pi,\quad
[P^t,\Pi]=D,
\\
[P^x,D]=P^x,\quad
[P^y,D]=P^y,\quad
[P^x,\Pi]=G^x,\quad
[P^y,\Pi]=G^y,
\\
[P^t,G^x]=P^x,\quad
[P^t,G^y]=P^y,\quad
[D,G^x]=G^x,\quad
[D,G^y]=G^y,
\\
[P^x,J]=P^y,\quad
[P^y,J]=-P^x,\quad
[G^x,J]=G^y,\quad
[G^y,J]=-G^x.
\end{gather*}

The Levi decomposition of the algebra~$\mathfrak g$ is
$
\mathfrak g=\langle P^t,D,\Pi\rangle\lsemioplus\langle J,P^x,P^y,G^x,G^y\rangle.
$
Here the subalgebra $\mathfrak f=\langle P^t,D,\Pi\rangle$ is a Levi factor of~$\mathfrak g$,
which is a~realization of the algebra ${\rm sl}(2,\mathbb R)$.
The radical~$\mathfrak r=\langle J,P^x,P^y,G^x,G^y\rangle$ of~$\mathfrak g$
is a~realization of an almost abelian algebra.
More specifically, $\mathfrak r=\mathfrak c\lsemioplus\mathfrak n$,
where $\mathfrak n=\langle P^x,P^y,G^x,G^y\rangle$ is the nilradical
(as well as the only maximal abelian ideal)
of both~$\mathfrak r$ and~$\mathfrak g$,
and the span $\mathfrak c=\langle J\rangle$ turns out a Cartan subalgebra of~$\mathfrak r$.
By~${\rm pr}_{\mathfrak f}$ and~${\rm pr}_{\mathfrak c}$ we denote the projectors defined by
the decomposition $\mathfrak g=\mathfrak f\lsemioplus(\mathfrak c\lsemioplus\mathfrak n)$.

The complete list of proper ideals of the algebra~$\mathfrak g$
is exhausted by the radical~$\mathfrak r$, the nilradical~$\mathfrak n$
and the derivative $\mathfrak g'=\mathfrak f\lsemioplus\mathfrak n$ of~$\mathfrak g$.
Each of these ideals is a  megaideal (i.e., a fully characteristic ideal)
of the algebra~$\mathfrak g$.
For~$\mathfrak g'$ this claim is obvious and
for~$\mathfrak r$ and~$\mathfrak n$ it follows from the fact that
the radical (resp.\ the nilradical) of a Lie algebra is
the unique maximal solvable (resp.\ nilpotent) ideal of this algebra
and hence this ideal is mapped by each automorphism of this algebra onto itself.

\begin{theorem}\label{thm:BurgersSystemCompletePointSymGroup}
The complete point symmetry group~$G$ of the Burgers system~\eqref{eq:BurgersSystem}
is generated by one-parameter groups associated with vector fields from the algebra~$\mathfrak g$
and one discrete transformation of simultaneous mirror mappings in the $(x,y)$- and $(u,v)$-planes,~e.g.,
\[
(t,x,y,u,v)\mapsto(t,-x,y,-u,v).
\]
\end{theorem}

\begin{proof}
We apply the automorphism-based version of the algebraic method for finding the complete point symmetry group
that involves factoring out inner automorphisms, cf.\ \cite{bihl2015a,card12a,hydo00b}.
Then we use constraints obtained by the algebraic method for components of point symmetry transformations
to complete the proof with the direct method.
See examples on computing point transformations between differential equations by the direct method,
e.g., in~\cite{king1998a}.

%\looseness=1
Each automorphism of~$\mathfrak g$ maps the Levi factor~$\mathfrak f$ of~$\mathfrak g$ onto a Levi factor of~$\mathfrak g$.
In view of the Levi--Malcev theorem, any two Levi factors of~$\mathfrak g$ are conjugated by an inner automorphism
generated by an element of the nilradical~$\mathfrak n$.
All inner automorphisms of~$\mathfrak g$ are induced by elements in the identity component of the group~$G$.
Hence one can assume that the Levi factor~$\mathfrak f$ is invariant
with respect to pushforwards by discrete point symmetries of the system~\eqref{eq:BurgersSystem}.
Since the Levi factor~$\mathfrak f$ is isomorphic to the algebra ${\rm sl}(2,\mathbb R)$,
it possesses a singe independent outer automorphism with the matrix $\mathop{\rm diag}(-1,1,-1)$.
As a result, to find discrete point symmetries of the system~\eqref{eq:BurgersSystem}
it suffices to consider automorphisms of~$\mathfrak g$ whose matrices are of the form
$\mathop{\rm diag}(\varepsilon,1,\varepsilon)\oplus\tilde A$,
where $\varepsilon=\pm1$ and $\tilde A$ is a $5\times5$ nondegenerate matrix.
The set of such automorphisms is a subgroup of the automorphism group of~$\mathfrak g$
and is exhausted by those with matrices of the form
\[
A=\mathop{\rm diag}(\varepsilon,1,\varepsilon,\varepsilon')
\oplus\varepsilon\begin{pmatrix}\varepsilon'a&-\varepsilon'b\\b&a\end{pmatrix}
\oplus\begin{pmatrix}\varepsilon'a&-\varepsilon'b\\b&a\end{pmatrix},
\]
where $\varepsilon,\varepsilon'=\pm1$, $(a,b)\ne(0,0)$,
and the parameter~$b$ can be set to be equal 0
using inner automorphisms associated with the basis vector field~$J$.
Thus, the final form of automorphism matrices to be considered is
\[
A=\mathop{\rm diag}(\varepsilon,1,\varepsilon,\varepsilon',\varepsilon\varepsilon'a,\varepsilon a,\varepsilon'a, a),
\]
where $\varepsilon,\varepsilon'=\pm1$ and $a\ne0$.

Suppose that the pushforward~$\mathcal T_*$ of vector fields by a point transformation
\[
\mathcal T\colon\quad (\tilde t,\tilde x,\tilde y,\tilde u,\tilde v)=(T,X,Y,U,V)(t,x,y,u,v)
\]
is the automorphism of~$\mathfrak g$ with the matrix~$A$, i.e.,
\begin{gather*}
\mathcal T_*P^t=\varepsilon\tilde P^t,\quad
\mathcal T_*P^x=\varepsilon\varepsilon'a\tilde P^x,\quad
\mathcal T_*P^y=\varepsilon a\tilde P^y,\quad
\mathcal T_*G^x=\varepsilon'a\tilde G^x,\quad
\mathcal T_*G^y=a\tilde G^y,\\
\mathcal T_*D  =\tilde D,\quad
\mathcal T_*\Pi=\varepsilon\tilde\Pi,\quad
\mathcal T_*J  =\varepsilon'\tilde J,
\end{gather*}
where tildes over vector fields mean that these vector fields are given in the new coordinates.
We componentwise split the above conditions for~$\mathcal T_*$
and thus derive a system of differential equations for the components of~$\mathcal T$,
\begin{gather*}
T_t=\varepsilon,\quad X_t=Y_t=U_t=V_t=0,\\
X_x=\varepsilon\varepsilon'a,\quad T_x=Y_x=U_x=V_x=0,\\
Y_y=\varepsilon a,\quad T_y=X_y=U_y=V_y=0,\\
tX_x+X_u=\varepsilon'aT,\quad U_u=\varepsilon'a,\quad T_u=Y_u=V_u=0,\\
tY_y+Y_v=aT,\quad V_v=a,\quad T_v=X_v=U_v=0,\\
tT_t=T,\quad xX_x-uX_u=X,\quad yY_y-vY_v=Y,\quad uU_u=U,\quad vV_v=V,\\
t^2T_t=\varepsilon T^2,\quad txX_x+(x-tu)X_u=\varepsilon TX,\quad tyY_y+(y-tv)Y_v=\varepsilon TY,\\ \qquad (x-tu)U_u=\varepsilon(X-TU),\quad (y-tv)V_v=\varepsilon(Y-TV),\\
yX_x+vX_u=\varepsilon'Y,\quad xY_y+uY_v=\varepsilon'X,\quad vU_u=\varepsilon'V,\quad uV_v=\varepsilon'U.
\end{gather*}
This system implies that $T=\varepsilon t$ and hence $X_u=Y_v=0$.
Furthermore, $X=\varepsilon\varepsilon'ax$, $Y=\varepsilon a y$, $U=\varepsilon'au$ and $V=av$.

Using the chain rule, we express all required transformed derivatives in terms of the initial coordinates
and substitute the obtained expressions into the copy of the Burgers system in the new coordinates.
The expanded system should vanish for each solution of the Burgers system.
This condition implies the equations $\varepsilon=1$ and $a^2=1$, i.e., $a=\pm1$.
The value $a=-1$ is mapped to the value $a=1$ by an inner automorphism related to the vector field~$J$.
Therefore, discrete symmetries of the Burgers system~\eqref{eq:BurgersSystem} are exhausted,
up to combining with continuous symmetries and with each other, by one involution, $(t,x,y,u,v)\to(t,-x,y,-u,v)$.
\end{proof}

\begin{corollary}
The factor group of the complete point symmetry group~$G$ of the Burgers system~\eqref{eq:BurgersSystem}
with respect to its identity component is isomorphic to the group $\mathbb Z_2$.
\end{corollary}

\begin{corollary}
The complete point symmetry group~$G$ of the Burgers system~\eqref{eq:BurgersSystem}
consists of the transformations of the form 
\begin{gather*}
\tilde t=\frac{\alpha t+\beta}{\gamma t+\delta},
\quad
\begin{pmatrix}\tilde x\\ \tilde y\end{pmatrix}
=\frac\sigma{\gamma t+\delta}O\begin{pmatrix}x\\y\end{pmatrix}
+\frac{\alpha t+\beta}{\gamma t+\delta}\begin{pmatrix}\mu_1\\ \mu_2\end{pmatrix}
+\begin{pmatrix}\nu_1\\ \nu_2\end{pmatrix},
\\[1ex]
\begin{pmatrix}\tilde u\\ \tilde v\end{pmatrix}
=\frac{\gamma t+\delta}{\sigma}O\begin{pmatrix}u\\v\end{pmatrix}
-\frac\gamma\sigma O\begin{pmatrix}x\\y\end{pmatrix}
+\begin{pmatrix}\mu_1\\ \mu_2\end{pmatrix},
\end{gather*}
where
$\alpha$, $\beta$, $\gamma$ and $\delta$ are arbitrary constants with $\alpha\delta-\beta\gamma>0$ 
such that their tuple is defined up to nonvanishing multiplier, $\sigma=\sqrt{\alpha\delta-\beta\gamma}$,
$O$ is an arbitrary $2\times2$ orthogonal matrix,
and $\mu_1$, $\mu_2$, $\nu_1$ and $\nu_2$ are arbitrary constants. 
\end{corollary}

\section{Symmetries of linearizable subset of solutions}\label{sec:SymsOfLinearizableSubsetOfSolutions}

It is possible to check with the infinitesimal invariance criterion and/or the chain rule
that all point symmetries of the Burgers system~\eqref{eq:BurgersSystem} are point symmetries
of the equation~\eqref{eq:BurgersSystemLinearizationConstraint}.
Therefore, the complete point symmetry group of the joint system~\eqref{eq:BurgersSystem},~\eqref{eq:BurgersSystemLinearizationConstraint},
which is further denoted by~$\mathcal S$,
contains the point symmetry group~$G$ of the system~\eqref{eq:BurgersSystem}.
Moreover, the following stronger assertion is also true.

\begin{proposition}\label{pro:BurgersSystemWithLinearizationConstraintLieSyms}
The maximal Lie invariance algebra and the complete point symmetry group of the system~$\mathcal S$,
which consists of the Burgers system~\eqref{eq:BurgersSystem}
jointly with the differential constraint~\eqref{eq:BurgersSystemLinearizationConstraint},
coincide with the algebra~$\mathfrak g$ and the group~$G$, respectively.
\end{proposition}

\begin{proof}
Since the system~$\mathcal S$ is overdetermined, we should carefully handle its differential consequences.
At the same time, there is the description of the solution set of~$\mathcal S$ in terms of
the generalized Hopf--Cole transformation~\eqref{eq:2DHopfColeTrans}
and the solution set of the (1+2)-dimensional linear heat equation~\eqref{eq:(1+2)DLinHeatEq}.
This description implies that independent differential consequences of~$\mathcal S$
whose orders as differential equations equal one or two
are exhausted by the equation~\eqref{eq:BurgersSystemLinearizationConstraint}
and the equations~\eqref{eq:BurgersSystem}, $u_{xy}=v_{xx}$ and $u_{yy}=v_{xy}$, respectively.
It is obvious that the last two equations are differential consequences of~\eqref{eq:BurgersSystemLinearizationConstraint}
that are obtained by single differentiations with respect to~$x$ and~$y$, respectively.

Therefore, for computing the maximal Lie invariance algebra of the system~$\mathcal S$
we need to first use the infinitesimal invariance criterion separately
for the equation~\eqref{eq:BurgersSystemLinearizationConstraint}.
Then, taking into account the derived constraints for components of Lie symmetry vector fields,
we apply the infinitesimal invariance criterion to the equations~\eqref{eq:BurgersSystem},
substitute for derivatives in view of all the above differential consequences
and split with respect to parametric derivatives.
The constructed system of determining equations for components of Lie symmetry vector fields
of the system~$\mathcal S$ has the same solution space as that for the Burgers system~\eqref{eq:BurgersSystem}.
In other words, the maximal Lie invariance algebra of the system~$\mathcal S$
coincides with the algebra~$\mathfrak g$.

Following the argumentation in the proof of Theorem~\ref{thm:BurgersSystemCompletePointSymGroup},
we see that the statement on the complete point symmetry group of the system~$\mathcal S$
is also true.
\end{proof}

Proposition~\ref{pro:BurgersSystemWithLinearizationConstraintLieSyms} means
that the Burgers system~\eqref{eq:BurgersSystem} possesses no genuine conditional symmetries
related to the differential constraint~\eqref{eq:BurgersSystemLinearizationConstraint}
although it possesses genuine potential conditional symmetries under this constraint.
Indeed, the equation~\eqref{eq:BurgersSystemLinearizationConstraint} is of conserved form.
Using it as a ``short'' conservation law, we introduce
the potential~$\psi$ defined by the equations $\psi_x=u$ and $\psi_y=v$.
The substitution of the expression of $(u,v)$ in terms of~$\psi$ reduces the system~\eqref{eq:BurgersSystem}
to the condition that the derivatives of $R=\psi_t+\tfrac12(\psi_x)^2+\tfrac12(\psi_y)^2-\psi_{xx}-\psi_{yy}$
with respect to~$x$ and~$y$ vanish,
i.e., the function~$R$ depends only on~$t$.
Since the potential~$\psi$ is defined up to summand being an arbitrary smooth function of~$t$,
we can make the function~$R$ to vanish,
which gives the equation for the potential~$\psi$,
\begin{gather}\label{eq:BurgersSystemEqLinearizableTo(1+2)DLinHeatEq}
\psi_t+\frac12(\psi_x)^2+\frac12(\psi_y)^2-\psi_{xx}-\psi_{yy}=0.
\end{gather}
This equation is linearized by a point transformation of the potential, $\phi=e^{-\psi/2}$,
to the equation~\eqref{eq:(1+2)DLinHeatEq}.
The expressions for~$u$ and~$v$ in terms of the modified potential~$\phi$ take the form~\eqref{eq:2DHopfColeTrans}.
Therefore, the equation~\eqref{eq:(1+2)DLinHeatEq} can interpreted as a potential equation
for the system~$\mathcal S$.
The maximal Lie invariance algebra~$\breve{\mathfrak g}$ of the equation~\eqref{eq:(1+2)DLinHeatEq}
is spanned by the vector fields
\begin{gather*}
\breve P^t=\p_t,\quad
\breve D=2t\p_t+x\p_x+y\p_y,\quad
\breve\Pi=4t^2\p_t+4tx\p_x+4ty\p_y-(x^2+y^2+4t)\phi\p_\phi,
\\
\breve J=x\p_y-y\p_x,\quad
\breve P^x=\p_x,\quad
\breve P^y=\p_y,\quad
\breve G^x=2t\p_x-x\phi\p_\phi,\quad
\breve G^y=2t\p_y-y\phi\p_\phi,
\\
\breve I=\phi\p_\phi,\quad
\breve Z(f)=f(t,x,y)\p_\phi,
\end{gather*}
where the parameter function $f=f(t,x,y)$ runs through the solution set of the equation~\eqref{eq:(1+2)DLinHeatEq}. 
Discrete symmetries of the equation~\eqref{eq:(1+2)DLinHeatEq} are exhausted,
up to combining with continuous symmetries and with each other, by two involution, 
$(t,x,y,\psi)\to(t,-x,y,\psi)$ and $(t,x,y,\psi)\to(t,x,y,-\psi)$.
The vector fields~$\breve P^t$, $\breve D$, $\breve\Pi$, $\breve J$, $\breve P^x$, $\breve P^y$, $\breve G^x$ and $\breve G^y$
induce, via the transformation~\eqref{eq:2DHopfColeTrans},
the basis elements~$P^t$, $D$, $\Pi$, $J$, $P^x$, $P^y$, $G^x$ and $G^y$ of the algebra~$\mathfrak g$, respectively.
In order to show this, it is necessary to prolong Lie symmetry vector fields of the equation~\eqref{eq:(1+2)DLinHeatEq}
to first derivatives of~$\phi$ and act by prolonged vector fields on $-2\phi_x/\phi$ and $-2\phi_y/\phi$.
Expressing the results of the action in terms of~$(u,v)$, when it is possible, gives the $u$- and $v$-components 
of the induced elements of the algebra~$\mathfrak g$.
In a similar way, the discrete symmetry $(t,x,y,\psi)\to(t,-x,y,\psi)$ of the equation~\eqref{eq:(1+2)DLinHeatEq} induces
the discrete symmetry $(t,x,y,u,v)\to(t,-x,y,-u,v)$ of the system~$\mathcal S$.
Therefore, the entire point symmetry group~$G$ of the system~$\mathcal S$ is induced
by the point symmetry group of the potential equation~\eqref{eq:(1+2)DLinHeatEq}.
The Lie symmetry vector field~$\breve I$ and the second independent discrete symmetry $(t,x,y,\psi)\to(t,x,y,-\psi)$ of the equation~\eqref{eq:(1+2)DLinHeatEq}
are mapped, in the above way, to the zero vector field and the identity transformation, respectively.
At the same time, vector fields from the infinite-dimensional ideal $\{\breve Z(f)\}$ of~$\breve{\mathfrak g}$,
which is related to the linear superposition of solutions of the equation~\eqref{eq:(1+2)DLinHeatEq},
have no counterparts among local infinitesimal symmetries of the system~$\mathcal S$
and thus are genuine potential symmetries of this system, i.e.,
the genuine conditional potential symmetries of the system~\eqref{eq:BurgersSystem}.

Another implication of Proposition~\ref{pro:BurgersSystemWithLinearizationConstraintLieSyms}
is that every Lie ansatz for the system~\eqref{eq:BurgersSystem} also reduces
the equation~\eqref{eq:BurgersSystemLinearizationConstraint}.
This is why for any one-dimensional subalgebra~$\mathfrak s$ of the algebra~$\mathfrak g$ 
or for any two-dimensional subalgebra~$\mathfrak s$ among $\mathfrak g^{2.1}_\kappa$--$\mathfrak g^{2.6}_\mu$,
the linearizable set of solutions of the form~\eqref{eq:2DHopfColeTrans} intersects
the set of $\mathfrak s$-invariant solutions of the system~\eqref{eq:BurgersSystem};
cf.\ Sections~\ref{sec:BurgersSystemLiereductionsOfCodim1} and~\ref{sec:BurgersSystemLieReductionOfCodim2}.
Reducing the system~\eqref{eq:BurgersSystem} by an ansatz for $\mathfrak s$-invariant solutions,
one should be interested only in solutions of the corresponding reduced system
that do not satisfy the reduced counterpart of~\eqref{eq:BurgersSystemLinearizationConstraint}.

\section{Subalgebras of Lie invariance algebra}\label{sec:SubalgebrasOfLieInvarianceAlgebra}

The classification of subalgebras of Galilei algebras was considered in a number of works,
see for example in \cite{bara1989a,bara1995a,FushchychBarannykBarannyk} and references therein.
We classified inequivalent one- or two-dimensional subalgebras of~$\mathfrak g$ from the very beginning
and compared the obtained list with the list presented in~\cite{FushchychBarannykBarannyk}.

We classify subalgebras of the algebra~$\mathfrak g$,
up to the equivalence relation generated by the adjoint action 
of the point symmetry group~$G$ of the Burgers system on~$\mathfrak g$.
The radical~$\mathfrak r$ and the nilradical~$\mathfrak n$
are megaideals (i.e., fully characteristic ideals) of~$\mathfrak g$
and hence they are $G$-invariant.
To characterize classification cases, for any subalgebra~$\mathfrak s$ of~$\mathfrak g$
we can introduce the $G^\sim$-invariant values
$\dim\mathfrak s\cap\mathfrak r$, $\dim\mathfrak s\cap\mathfrak n$,
$\dim{\rm pr}_{\mathfrak f}\mathfrak s$ and $\dim{\rm pr}_{\mathfrak c}\mathfrak s$.

For efficiently recognizing inequivalent subalgebras of~$\mathfrak g$,
we consider their projections on the Levi factor~$\mathfrak f$.
These projections are necessarily subalgebras of~$\mathfrak f$,
and, moreover, the projections of equivalent subalgebras of~$\mathfrak g$ are equivalent as subalgebras of~$\mathfrak f$.
A~complete list of inequivalent subalgebras of the algebra ${\rm sl}(2,\mathbb R)$ is well known;
see, e.g.,~\cite{pate1977a}.
In terms of the realization~$\mathfrak f$ of~${\rm sl}(2,\mathbb R)$, it is exhausted by
$\{0\}$, $\langle P^t\rangle$, $\langle D\rangle$, $\langle P^t+\Pi\rangle$,
$\langle P^t, D\rangle$ and~$\mathfrak f$ itself.
Considering each of the listed subalgebras of~$\mathfrak f$
as a projection of an appropriate subalgebra of~$\mathfrak g$,
we try to add elements of the radical~$\mathfrak r$ to the basis elements of this subalgebra,
and to additionally extend the basis by elements from the radical~$\mathfrak r$.

A complete list of one-dimensional $G$-inequivalent subalgebras of~$\mathfrak g$ is exhausted by the subalgebras
\begin{gather*}
\mathfrak g^{1.1}_\kappa=\langle P^t+\kappa J          \rangle_{\kappa\in\{0,1\}},\quad
\mathfrak g^{1.2}       =\langle P^t+G^y               \rangle,                   \quad
\mathfrak g^{1.3}_\kappa=\langle D+2\kappa J           \rangle_{\kappa\geqslant0},\\[.5ex]
\mathfrak g^{1.4}_\kappa=\langle P^t+\Pi+\kappa J      \rangle_{\kappa\geqslant0},\quad
\mathfrak g^{1.5}_\mu   =\langle P^t+\Pi+J+\mu(G^x-P^y)\rangle_{\mu>0},           \\[.5ex]
\mathfrak g^{1.6}       =\langle J                     \rangle,\quad
\mathfrak g^{1.7}       =\langle G^x-P^y               \rangle,\quad
\mathfrak g^{1.8}       =\langle P^y                   \rangle.
\end{gather*}

A similar list of two-dimensional $G$-inequivalent subalgebras consists of the subalgebras
\begin{gather*}
\mathfrak g^{2.1}_\kappa  =\langle P^t,D+\kappa J                \rangle_{\kappa\geqslant0},\quad
\mathfrak g^{2.2}         =\langle P^t,J                         \rangle,\quad
\mathfrak g^{2.3}         =\langle D,J                           \rangle,\quad
\mathfrak g^{2.4}         =\langle P^t+\Pi,J                     \rangle,\\[.5ex]
\mathfrak g^{2.5}_\mu     =\langle P^t+\Pi+J+\mu(G^y+P^x),G^x-P^y\rangle_{\mu\geqslant0},\quad
%\mathfrak g^{2.6}_{\kappa}=\langle G^x-P^y,G^y+P^x+\kappa P^y    \rangle_{|\kappa|<2},\\ %%% Barranyks' version
\mathfrak g^{2.6}_\mu     =\langle G^x-P^y,G^y+\mu P^x           \rangle_{\mu>0},\\[.5ex]
\mathfrak g^{2.7}_{\mu\nu}=\langle P^y,P^t+\mu G^x+\nu G^y       \rangle_{\mu,\nu\geqslant0,\ \mu^2+\nu^2\in\{0,1\}},\quad
\mathfrak g^{2.8}         =\langle P^y,D                         \rangle,\\[.5ex]
\mathfrak g^{2.9}         =\langle P^y,P^x                       \rangle,\quad\!
\mathfrak g^{2.10}        =\langle P^y,G^y                       \rangle,\quad\!
\mathfrak g^{2.11}_\mu    =\langle P^y,G^x+\mu G^y               \rangle_{\mu\geqslant0},\quad\!
\mathfrak g^{2.12}        =\langle P^y,G^y+P^x                   \rangle.
\end{gather*}

\section{Lie reductions of codimension one}\label{sec:BurgersSystemLiereductionsOfCodim1}

Ansatzes constructed with one-dimensional subalgebras of~$\mathfrak g$
reduce the system~\eqref{eq:BurgersSystem} to systems of two partial differential equations
in two independent variables.
Below for each of the one-dimensional subalgebras listed in the previous section,
we present an ansatz constructed for $(u,v)$
and the corresponding reduced system.
Here $w^i=w^i(z_1,z_2)$, $i=1,2$,
are new unknown functions of the invariant independent variables $(z_1,z_2)$.
The subscripts~1 and~2 of~$w$'s denote derivatives with respect to~$z_1$ and~$z_2$, respectively.
We assume summation with respect to the repeated index~$i$.

\bigskip\par\noindent
1.1. $\mathfrak g^{1.1}_\kappa=\langle P^t+\kappa J\rangle_{\kappa\in\{0,1\}}$:
\begin{gather*}
u=w^1\cos\tau-w^2\sin\tau-\kappa y, \\
v=w^1\sin\tau+w^2\cos\tau+\kappa x,
\end{gather*}
where\quad
$z_1= x\cos\tau+y\sin\tau$,\quad
$z_2=-x\sin\tau+y\cos\tau$,\quad
$\tau:=\kappa t$;
\begin{gather*}
w^iw^1_i-w^1_{ii}-2\kappa w^2-\kappa z_1=0,\\
w^iw^2_i-w^2_{ii}+2\kappa w^1-\kappa z_2=0.
\end{gather*}

\medskip\par\noindent
1.2. $\mathfrak g^{1.2}=\langle P^t+G^y\rangle$:\quad
$u=w^1, \quad v=w^2+t,$\quad
where \quad
\smash{$z_1=x, \quad z_2=y-\dfrac{t^2}{2}$};
\begin{gather*}
w^iw^1_i-w^1_{ii}=0,\\
w^iw^2_i-w^2_{ii}+1=0.
\end{gather*}

\medskip\par\noindent
1.3. $\mathfrak g^{1.3}_\kappa=\langle D+2\kappa J\rangle_{\kappa\geqslant0}$:
\begin{gather*}
u=\frac1{\sqrt{|t|}}(w^1\cos\tau-w^2\sin\tau)+\frac x{2t}-\kappa\frac yt, \\
v=\frac1{\sqrt{|t|}}(w^1\sin\tau+w^2\cos\tau)+\frac y{2t}+\kappa\frac xt,
\end{gather*}
where\quad
$z_1=\dfrac1{\sqrt{|t|}}( x\cos\tau+y\sin\tau)$,\quad
$z_2=\dfrac1{\sqrt{|t|}}(-x\sin\tau+y\cos\tau)$,\quad
$\tau:=\kappa\ln|t|$;\quad
($\hat\kappa:=\kappa\mathop{\rm sgn}\nolimits t$)
\begin{gather*}
w^iw^1_i-w^1_{ii}-2\hat\kappa w^2-\left(\kappa^2+\frac14\right)z_1=0,\\
w^iw^2_i-w^2_{ii}+2\hat\kappa w^1-\left(\kappa^2+\frac14\right)z_2=0.
\end{gather*}

\medskip\par\noindent
1.4. $\mathfrak g^{1.4}_\kappa=\langle P^t+\Pi+\kappa J\rangle_{\kappa\geqslant0}$:
\begin{gather*}
u=\frac1{\sqrt{t^2+1}}(w^1\cos\tau-w^2\sin\tau)+\frac{tx}{t^2+1}-\frac{\kappa y}{t^2+1},\\
v=\frac1{\sqrt{t^2+1}}(w^1\sin\tau+w^2\cos\tau)+\frac{ty}{t^2+1}+\frac{\kappa x}{t^2+1},
\end{gather*}
where\quad
$z_1=\dfrac1{\sqrt{t^2+1}}( x\cos\tau+y\sin\tau)$,\quad
$z_2=\dfrac1{\sqrt{t^2+1}}(-x\sin\tau+y\cos\tau)$,\quad
$\tau:=\kappa\tan^{-1}t$;
\begin{gather*}
w^iw^1_i-w^1_{ii}-2\kappa w^2+(1-\kappa^2)z_1=0,\\
w^iw^2_i-w^2_{ii}+2\kappa w^1+(1-\kappa^2)z_2=0.
\end{gather*}

\medskip\par\noindent
1.5. $\mathfrak g^{1.5}_{\mu}=\langle P^t+\Pi+J+\mu(G^x-P^y)\rangle_{\mu>0}$:
\begin{gather*}
u=\frac{tw^1+w^2}{t^2+1}+\frac{t(x+\mu)}{t^2+1}-\frac{y}{t^2+1},\\
v=\frac{-w^1+tw^2}{t^2+1}+\frac{ty}{t^2+1}+\frac{x-\mu}{t^2+1},
\end{gather*}
where\quad
$z_1=\dfrac{tx-y}{t^2+1}-\mu\arctan t$, \quad
$z_2=\dfrac{x+ty}{t^2+1}$;
\begin{gather*}
w^iw^1_i-w^1_{ii}-2w^2=0,\\
w^iw^2_i-w^2_{ii}+2w^1+2\mu=0.
\end{gather*}

\noprint{%%%%%%%%%%%%%%%%%%%%%%%%%%%% OLD VERSION OF ANSATZ
\medskip\par\noindent
1.6. $\mathfrak g^{1.6}=\langle J\rangle$:\quad
$u=\dfrac xrw^1-\dfrac yrw^2$,\quad
$v=\dfrac yrw^1+\dfrac xrw^2$,\quad
where\quad
$z_1=t$, \quad
$z_2=r:=\sqrt{x^2+y^2}$;
\\[1ex](another version of the ansatz can be obtained by replacing
$(w^1,w^2)$ by $(w^1+1/r,w^2/r)$ );
\begin{gather*}
w^1_1+w^1w^1_2-\frac{(w^2)^2}{z_2}-w^1_{22}-\frac{w^1_2}{z_2}+\frac{w^1}{z_2^2}=0,\\
w^2_1+w^1w^2_2+\frac{w^1w^2}{z_2}-w^2_{22}-\frac{w^2_2}{z_2}+\frac{w^2}{z_2^2}=0.
\end{gather*}
}%%%%%%%%%%%%%%%%%%%%%%%%%%%%%%%%%%%%%

\medskip\par\noindent
1.6. $\mathfrak g^{1.6}=\langle J\rangle$:\quad
$u=\dfrac xrw^1-\dfrac yrw^2+\dfrac x{r^2}$,\quad
$v=\dfrac yrw^1+\dfrac xrw^2+\dfrac y{r^2}$,\\[1ex]
where \quad
$z_1=t$, \quad
$z_2=r:=\sqrt{x^2+y^2}$;
\begin{gather*}
w^1_1+w^1w^1_2-w^1_{22}-\frac{(w^2)^2}{z_2}-\frac1{z_2^3}=0,\\
w^2_1+w^1w^2_2-w^2_{22}+\frac{w^1w^2}{z_2}+2\frac{w^2}{z_2^2}=0.
\end{gather*}

\medskip\par\noindent
1.7. $\mathfrak g^{1.7}=\langle G^x-P^y\rangle$:\quad
$u=\dfrac{w^1-tw^2+tx-y}{t^2+1}$,\quad
$v=\dfrac{tw^1+w^2+x+ty}{t^2+1}$,\\
where\quad
$z_1=\arctan t$,\quad
$z_2=\dfrac{x+ty}{t^2+1}$;
\begin{gather*}
w^1_1+w^1w^1_2-w^1_{22}-2w^2=0,\\
w^2_1+w^1w^2_2-w^2_{22}+2w^1=0.
\end{gather*}

\medskip\par\noindent
1.8. $\mathfrak g^{1.8}=\langle P^y\rangle$:\quad
$u=w^1$,\quad
$v=w^2$,\quad
where \quad
$z_1=t$,\quad
$z_2=x$;
\begin{gather*}
w^1_1+w^1w^1_2-w^1_{22}=0,\\
w^2_1+w^1w^2_2-w^2_{22}=0.
\end{gather*}

The differential constraint $u_y=v_x$, which is needed for the linearizability of the system~\eqref{eq:BurgersSystem},
is respectively reduced by the above ansatzes to the following differential constraints
in terms of invariant variables:
\begin{gather*}
1.1.\ w^1_2=w^2_1+2\kappa,\quad
1.2.\ w^1_2=w^2_1,\quad
1.3.\ w^1_2=w^2_1+2\kappa\sgn t,\quad
1.4.\ w^1_2=w^2_1+2\kappa,\\
1.5.\ w^1_2=w^2_1+2,\quad
1.6.\ z_2w^2_2+w^2=0,\quad
1.7.\ w^2_2+2=0,\quad
1.8.\ w^2_2=0.\quad
\end{gather*}
For each of the above reduced systems,
only solutions that do not satisfy the differential constraint with the same number
are essential for finding exact solutions of the system~\eqref{eq:BurgersSystem}.

\section{Symmetry analysis of reduced systems of PDEs}\label{sec:SymmetryAnalysisOfReducedSystemsOfPDEs}

We have selected ansatzes in such a way that
the reduced systems are of quite simple form and can be grouped into two sets
depending on their structure,
which is convenient for studying their symmetries and finding exact solutions.

Thus, reduced systems 1.1--1.5 are of the form
\begin{gather*}
w^iw^1_i-w^1_{ii}-2\kappa w^2+\alpha z_1=0,\\
w^iw^2_i-w^2_{ii}+2\kappa w^1+\alpha z_2+\beta=0,
\end{gather*}
where $\kappa$, $\alpha$ and~$\beta$ are constants with $\alpha\beta=0$.
Depending on values of these parameters,
a system of the above form admits the following maximal Lie invariance algebra~$\mathfrak a$:
\[\arraycolsep=0ex
\begin{array}{ll}
\alpha\ne0,\ \beta=0        \colon&\quad\mathfrak a=\langle\tilde J                               \rangle,\\[.5ex]
\alpha=0,\ \beta\ne0        \colon&\quad\mathfrak a=\langle\tilde P^1,\tilde P^2                  \rangle,\\[.5ex]
\alpha=\beta=0,\ \kappa\ne0 \colon&\quad\mathfrak a=\langle\tilde P^1,\tilde P^2,\tilde J         \rangle,\\[.5ex]
\alpha=\beta=\kappa=0       \colon&\quad\mathfrak a=\langle\tilde P^1,\tilde P^2,\tilde J,\tilde D\rangle.
\end{array}
\]
Here we denote
\begin{gather*}
\tilde P^1=\p_{z_1},\quad
\tilde P^2=\p_{z_2},\quad
\tilde J=z_1\p_{z_2}-z_2\p_{z_1}+w^1\p_{w^2}-w^2\p_{w^1},\\
\tilde D=z_1\p_{z_1}+z_2\p_{z_2}-w^1\p_{w^1}-w^2\p_{w^2}.
\end{gather*}
As a result, the maximal Lie invariance algebras of reduced systems~1.1--1.5 are respectively
\begin{gather*}
\mathfrak a^1=\langle\tilde J\rangle \quad\mbox{if}\quad \kappa=1 \qquad\mbox{and}\qquad
\mathfrak a^1=\langle\tilde P^1,\tilde P^2,\tilde J,\tilde D\rangle \quad\mbox{if}\quad \kappa=0;
\\
\mathfrak a^2=\langle\tilde P^1,\tilde P^2\rangle;
\qquad
\mathfrak a^3=\langle\tilde J\rangle;
\\
\mathfrak a^4=\langle\tilde J\rangle \quad\mbox{if}\quad \kappa\ne1 \qquad\mbox{and}\qquad
\mathfrak a^4=\langle\tilde P^1,\tilde P^2,\tilde J\rangle \quad\mbox{if}\quad \kappa=1;
\qquad
\mathfrak a^5=\langle\tilde P^1,\tilde P^2\rangle.
\end{gather*}

The other reduced systems, 1.6--1.8, are of the form
\begin{gather*}
w^j_1+w^1w^j_2-w^j_{22}+F^j(z_2,w^1,w^2)=0,\quad j=1,2,
\end{gather*}
where the parameter functions $F^j=F^j(z_1,z_2,w^1,w^2)$ are at most quadratic in $(w^1,w^2)$.
The maximal Lie invariance algebras of these reduced systems are of different structure:
\begin{gather*}
\mathfrak a^6=\langle
\p_{z_1},\
2z_1\p_{z_1}+z_2\p_{z_2}-w^1\p_{w^1}-w^2\p_{w^2},\\\phantom{\mathfrak a^6=\langle}
z_1^2\p_{z_1}+z_1z_2\p_{z_2}+(z_2-z_1w^1)\p_{w^1}-z_1w^2\p_{w^2}
\rangle,
\\[1ex]
\mathfrak a^7=\langle
\p_{z_1},\
\p_{z_2},\
\cos(2z_1)\p_{z_2}-2\sin(2z_1)\p_{w^1}-2\cos(2z_1)\p_{w^2},\\\phantom{\mathfrak a^7=\langle}
\sin(2z_1)\p_{z_2}+2\cos(2z_1)\p_{w^1}-2\sin(2z_1)\p_{w^2}
\rangle,
\\[1ex]
\mathfrak a^8=\langle
\p_{z_1},\
2z_1\p_{z_1}+z_2\p_{z_2}-w^1\p_{w^1},\
z_1^2\p_{z_1}+z_1z_2\p_{z_2}+(z_2-z_1w^1)\p_{w^1},\\\phantom{\mathfrak a^8=\langle}
\p_{z_2},\
z_1\p_{z_2}+\p_{w^1},\
\p_{w^2},\
w^2\p_{w^2},\
w^1\p_{w^2},\
(z_2-z_1w^1)\p_{w^2}
\rangle.
\end{gather*}

If a~reduced system possesses Lie symmetries
that are not induced by Lie symmetries of an original system,
then the original system is said to admit \emph{additional}~\cite{olve1993b}
(or \emph{hidden}~\cite{abra2008a}) symmetries
with respect to the corresponding reduction.
The first example of such symmetries was constructed in~\cite{kapi1978a}
for the axisymmetric reduction of the incompressible Euler equations,
and this example was discussed in~\cite[Example~3.5]{olve1993b}.
Hidden symmetries of the Navier--Stokes equations were comprehensively studied in~\cite{fush1994a,fush1994b}.

In order to clarify which Lie symmetries of reduced systems~1.1--1.8 are induced
by Lie symmetries of the original Burgers system~\eqref{eq:BurgersSystem},
for each $m\in\{1,\dots,8\}$ we compute the normalizer of the subalgebra~$\mathfrak g^{1.m}$
in the algebra~$\mathfrak g$,
\[{\rm N}_{\mathfrak g}(\mathfrak g^{1.m})=\{Q\in\mathfrak g\mid [Q,Q']\in\mathfrak g^{1.m}\ \mbox{for all}\ Q'\in\mathfrak g^{1.m}\}.\]
The algebra of induced Lie symmetries of reduced system~1.$m$
is isomorphic to the quotient algebra ${\rm N}_{\mathfrak g}(\mathfrak g^{1.m})/\mathfrak g^{1.m}$.
Therefore, all Lie symmetries of reduced system~1.$m$
are induced by Lie symmetries of the original Burgers system~\eqref{eq:BurgersSystem}
if and only if  $\dim\mathfrak a^m=\dim{\rm N}_{\mathfrak g}(\mathfrak g^{1.m})-1$.
Thus,
\begin{gather*}
{\rm N}_{\mathfrak g}(\mathfrak g^{1.1}_\kappa)=\langle P^t,J\rangle\mbox{ \ if \ } \kappa=1 \quad\mbox{and}\quad
{\rm N}_{\mathfrak g}(\mathfrak g^{1.1}_\kappa)=\langle P^t,D,J,P^x,P^y\rangle \mbox{ \ if \ } \kappa=0,
\\
{\rm N}_{\mathfrak g}(\mathfrak g^{1.2})=\langle P^t+G^y,P^x,P^y\rangle,\quad
{\rm N}_{\mathfrak g}(\mathfrak g^{1.3}_\kappa)=\langle D,J\rangle,
\\
{\rm N}_{\mathfrak g}(\mathfrak g^{1.4}_\kappa)=\langle P^t+\Pi,J\rangle
\ \mbox{if}\ \kappa\ne1 \quad\mbox{and}\quad
{\rm N}_{\mathfrak g}(\mathfrak g^{1.4}_\kappa)=\langle P^t+\Pi,J,G^x-P^y,G^y+P^x\rangle \ \mbox{if}\ \kappa=1,
\\
{\rm N}_{\mathfrak g}(\mathfrak g^{1.5}_\mu)=\langle P^t+\Pi+J,G^x-P^y,G^y+P^x\rangle,\quad
{\rm N}_{\mathfrak g}(\mathfrak g^{1.6})=\langle P^t,D,\Pi,J\rangle,
\\
{\rm N}_{\mathfrak g}(\mathfrak g^{1.7})=\langle P^t+\Pi+J,P^x,P^y,G^x,G^y\rangle,\quad
{\rm N}_{\mathfrak g}(\mathfrak g^{1.8})=\langle P^t,D,P^x,P^y,G^x,G^y\rangle.
\end{gather*}
Comparing the dimensions of ${\rm N}_{\mathfrak g}(\mathfrak g^{1.m})$ and~$\mathfrak a^m$,
we conclude that
all Lie symmetries of reduced systems~1.1--1.7 are induced by Lie symmetries
of the original Burgers system~\eqref{eq:BurgersSystem}
but this is not the case for reduced system~1.8.
Therefore, the study of further Lie reductions of reduced systems~1.1--1.7 to systems of ODEs is needless
since it is more efficient to directly reduce the system~\eqref{eq:BurgersSystem} to systems of ODEs
using two-dimensional subalgebras of the algebra~$\mathfrak g$,
which is done in Section~\ref{sec:BurgersSystemLieReductionOfCodim2}.
In general, each direct reduction of codimension two corresponds to several two-step reductions.

We denote basis elements of~the algebra~$\mathfrak g^{1.8}$ by
$\hat P^1$, $\hat D$, $\hat\Pi$, $\hat P^2$, $\hat G$, $\hat G'$, $\hat I^{22}$, $\hat I^{21}$ and $\hat S$
following the order in which these elements are listed above.
Since the expressions of invariant variables for the algebra~$\mathfrak g^{1.8}$ are especially simple,
$z_1=t$, $z_2=x$, $w^1=u$ and $w^2=v$,
it is obvious that the elements of ${\rm N}_{\mathfrak g}(\mathfrak g^{1.8})$ induce the subalgebra
$\mathfrak s=\langle
\hat P^1, %\p_{z_1},\
\hat D  , %2z_1\p_{z_1}+z_2\p_{z_2}-w^1\p_{w^1}-w^2\p_{w^2},\
\hat P^2, %\p_{z_2},\
\hat G  , %z_1\p_{z_2}+\p_{w^1},\
\hat G'   %\p_{w^2}
\rangle$
of~$\mathfrak a^8$.
All vector fields from the complement of~$\mathfrak s$ in~$\mathfrak a^8$
are hidden symmetries of the system~\eqref{eq:BurgersSystem}
that are associated with the subalgebra~$\mathfrak g^{1.8}$ of the Lie invariance algebra~$\mathfrak g$.

Reduced system~1.8 is also singular from the point of view of other properties;
cf.\ Section~\ref{sec:BurgersSystemCLs}.
This system is partially coupled.
Its first equation is only in~$w^1$ and coincides with the classical Burgers equation,
and its second equation is linear with respect to~$w^2$.
These two facts give us the hint that reduced system~1.8 can be linearized
by a substitution related to the Hopf--Cole transformation.
Indeed, substituting
\begin{gather}\label{eq:Hopf--Cole-likeTransForReducedSystem1.8}
w^1=-2\frac{\tilde w^1_2}{\tilde w^1}, \quad
w^2=\frac{\tilde w^2}{\tilde w^1}
\end{gather}
into reduced system~1.8, we obtain
\[
-2\left(\frac{\tilde w^1_1-\tilde w^1_{22}}{\tilde w^1}\right)_2=0,\quad
\frac{\tilde w^1(\tilde w^2_1-\tilde w^2_{22})-\tilde w^2(\tilde w^1_1-\tilde w^1_{22})}{(\tilde w^1)^2}=0.
\]
\looseness=-1
We follow the standard procedure of applying the Hopf--Cole transformation
and integrate once the first of the obtained equations with respect to~$z_2$.
This results in the equation $\tilde w^1_1-\tilde w^1_{22}=h\tilde w^1$,
where $h=h(z_1)$ is an arbitrary smooth function of~$z_1=t$.
Since the new unknown function~$\tilde w^1$ in the Hopf--Cole transformation
is defined up to a multiplier being an arbitrary nonvanishing smooth function of~$z_1$,
we can set the function~$h$ to be identically equal to~$0$, i.e., $\tilde w^1_1-\tilde w^1_{22}=0$.
Then the second obtained equation implies $\tilde w^2_1-\tilde w^2_{22}=0$.
Summing up, reduced system~1.8 is linearized by the substitution~\eqref{eq:Hopf--Cole-likeTransForReducedSystem1.8}
to the decoupled system of two copies of the (1+1)-dimensional linear heat equation,
\begin{gather}\label{eq:SystemOfTwoLinHeatEqs}
\tilde w^j_1=\tilde w^j_{22},\quad j=1,2.
\end{gather}
As a result, we can construct wide families of exact invariant solutions of the Burgers system~\eqref{eq:BurgersSystem}
using ansatz~1.8, the substitution~\eqref{eq:Hopf--Cole-likeTransForReducedSystem1.8} and
known exact solutions of the (1+1)-dimensional linear heat equation.
This is why the situation for Lie reduction~1.8 is opposite to that for Lie reductions~1.1--1.7:
all Lie reductions of the system~\eqref{eq:BurgersSystem} with respect to two-dimensional subalgebras of the algebra~$\mathfrak g$
that are equivalent to two-step Lie reductions with reduction~1.8 as the first step
should be excluded from the consideration as needless.

$\mathfrak g^{1.8}$-invariant solutions of the system~\eqref{eq:BurgersSystem} satisfy
the differential constraint $u_y=v_y=0$ that differs from the differential constraint $u_y=v_x$.
Therefore, we obtain one more linearizable subset of solutions of the system~\eqref{eq:BurgersSystem}
although it is narrower than the subset singled out by the differential constraint $u_y=v_x$
since it is parameterized by two solutions of the (1+1)-dimensional linear heat equation
in contrast to the latter subset parameterized by one solution of the (1+2)-dimensional linear heat equation.
At the same time, the former subset can be extended by transformations from~$G$
that do not preserve~$\mathfrak g^{1.8}$,
whereas the latter subset is $G$-invariant.
It is obvious that the intersection of these subsets is associated with the differential constraint $u_y=v_x=v_y=0$
This means in terms of the invariants of the subalgebra~$\mathfrak g^{1.8}$
that $w^2$ is a constant, i.e., $\tilde w^2$ is proportional to~$\tilde w^1$.

The Lie invariance algebra~$\mathfrak p$ of the system~\eqref{eq:SystemOfTwoLinHeatEqs}
is much wider than that of reduced system~1.8.
It is spanned by the vector fields
\begin{gather*}
\tilde P^1=\p_{z_1},\quad
\tilde D  =2z_1\p_{z_1}+z_2\p_{z_2},\\
\tilde \Pi=4z_1^2\p_{z_1}+4z_1z_2\p_{z_2}-(z_2^2+2z_1)\tilde w^1\p_{\tilde w^1}-(z_2^2+2z_1)\tilde w^2\p_{\tilde w^2},\\
\tilde P^2=\p_{z_2},\quad
\tilde G  =2z_1\p_{z_2}-z_2\tilde w^1\p_{\tilde w^1}-z_2\tilde w^2\p_{\tilde w^2},\quad
\tilde I^{ij}=\tilde w^j\p_{\tilde w^i},\quad
f^i(z_1,z_2)\p_{\tilde w^i},\quad i,j=1,2,
%f^1(z_1,z_2)\p_{\tilde w^1},\quad f^2(z_1,z_2)\p_{\tilde w^2},
\end{gather*}
where $f^1=f^1(z_1,z_2)$ and $f^2=f^2(z_1,z_2)$ run through the solution set of the linear heat equation $f_1=f_{22}$.
The single linear heat equation admits the two independent recursion operators 
$\mathcal R_1=\mathrm D_2$ and $\mathcal R_2=2z_1\mathrm D_2+z_2$ with commutation relation $\mathcal R_1\mathcal R_2-\mathcal R_2\mathcal R_1=1$
\cite[Example~5.21]{olve1993b}.
Here $\mathrm D_2$ denotes the operator of total derivative with respect to the variable $z_2$;
cf.\ Section~\ref{sec:BurgersSystemCLs}.
Since the system~\eqref{eq:SystemOfTwoLinHeatEqs} consists of two copies of the linear heat equation 
that are not coupled to each other, it possesses the eight recursion operators $E_{ij}\mathcal R_1$ and $E_{ij}\mathcal R_2$,
where $E_{ij}$ denotes the $2\times 2$ matrix with unit in the $i$th row and the $j$th column and with zero otherwise,
$i,j=1,2$.
Generalized symmetries of the system~\eqref{eq:SystemOfTwoLinHeatEqs} 
can be easily constructed by its recursion operators acting iteratively on the characteristic $(\tilde w^1,\tilde w^2)$
of the Lie symmetry vector field $\tilde w^1\p_{\tilde w^1}+\tilde w^2\p_{\tilde w^2}$.
Since the system~\eqref{eq:SystemOfTwoLinHeatEqs} can be interpreted
as a potential system of reduced system~1.8 (see again Section~\ref{sec:BurgersSystemCLs}),
Lie (resp. generalized) symmetries of the system~\eqref{eq:SystemOfTwoLinHeatEqs}
can be interpreted as potential (resp. potential generalized) symmetries of reduced system~1.8
and hence as hidden potential (resp. hidden potential generalized) symmetries of the original Burgers system~\eqref{eq:BurgersSystem}.
In fact, some of potential symmetries of reduced system~1.8 correspond to its Lie symmetry.
In order to find this correspondence, we prolong elements of the algebra~$\mathfrak p$ to $\tilde w^1_2$
and then, using~\eqref{eq:Hopf--Cole-likeTransForReducedSystem1.8}, to~$w^1$ and~$w^2$.
Each vector field from~$\mathfrak p$ whose prolongation is projectable to the space of $(z_1,z_2,w^1,w^2)$
induces an element of~$\mathfrak a^8$, which just coincides with the projection of the prolonged vector field.
A similar procedure can be applied to generalized symmetries of the system~\eqref{eq:SystemOfTwoLinHeatEqs}.
Thus,%
\noprint{
\begin{gather*}
\tilde P^1   \to \hat P^1   ,\quad
\tilde D     \to \hat D     ,\quad
\tilde \Pi   \to4\hat \Pi   ,\quad
\tilde P^2   \to \hat P^2   ,\quad
\tilde G     \to2\hat G     ,\quad
\tilde I^{21}\to \hat P^1   ,\quad
\tilde I^{22}\to \hat I^{22}.
\end{gather*}
}
$\tilde P^1   \to \hat P^1   $,
$\tilde D     \to \hat D     $,
$\tilde \Pi   \to4\hat \Pi   $,
$\tilde P^2   \to \hat P^2   $,
$\tilde G     \to2\hat G     $,
$\tilde I^{21}\to \hat P^1   $,
$\tilde I^{22}\to \hat I^{22}$.
The basis element~$\tilde I^{11}$ of~$\mathfrak p$ is mapped to 0,
and the basis element~$\tilde I^{12}$ induces the generalized symmetry $-2w^2_2\p_{w^1}-(w^2)^2\p_{w^2}$ of reduced system~1.8.
Vector fields of the form $f^i\p_{w^i}$, which are associated with linear superposition of solutions of the system~\eqref{eq:SystemOfTwoLinHeatEqs},
have no counterparts among local infinitesimal symmetries of reduced system~1.8.
The last two basis elements~$\hat I^{21}$ and~$\hat S$ of the algebra~$\mathfrak a^8$
are induced by the generalized symmetries $-2\tilde w^1_2\p_{\tilde w^2}$ and $(2z_1\tilde w^1_2+2\tilde w^1)\p_{\tilde w^2}$
of the system~\eqref{eq:SystemOfTwoLinHeatEqs},
which are the results of acting by the recursion operators $-2E_{21}\mathcal R_1$ and $E_{21}\mathcal R_2$
on the Lie symmetry vector field $\tilde w^1\p_{\tilde w^1}+\tilde w^2\p_{\tilde w^2}$, respectively.
This gives an additional justification that the study of Lie reductions of reduced system~1.8 is needless.

\section{Lie reductions of codimension two}\label{sec:BurgersSystemLieReductionOfCodim2}

Since the reduced system constructed with the subalgebra $\mathfrak g^{1.8}=\langle P^y\rangle$ is linearizable,
a two-dimensional subalgebra of~$\mathfrak g$ is essential
for using in the course of Lie reduction of the system~\eqref{eq:BurgersSystem}
only if it does not contain the vector field~$P^y$
or, more generally, a vector field $G$-equivalent to~$P^y$,
which is a $G$-invariant property.
Note that the subalgebra $\mathfrak g^{2.10}$ is totally not appropriate for using within the framework of Lie reduction.
Therefore, only the subalgebras $\mathfrak g^{2.1}_\kappa$--$\mathfrak g^{2.6}_\mu$
are essential for Lie reduction among the two-dimensional listed inequivalent subalgebras.
Below for each of these subalgebras,
we present an ansatz constructed for $(u,v)$ and the corresponding reduced system.
Here $\varphi^i=\varphi^i(\omega)$, $i=1,2$,
are new unknown functions of the invariant independent variable~$\omega$,
and $r:=\sqrt{x^2+y^2}$.

\bigskip\par\noindent
2.1. $\mathfrak g^{2.1}_\kappa =\langle P^t,D+\kappa J\rangle_{\kappa\geqslant0}$:\\[1ex]
$u=\dfrac x{r^2}\varphi^1-\dfrac y{r^2}\varphi^2$,\quad
$v=\dfrac y{r^2}\varphi^1+\dfrac x{r^2}\varphi^2$,\quad
where\quad
$\omega=\arctan\dfrac yx-\kappa\ln r$;
\begin{gather*}
(\varphi^2-\kappa\varphi^1-2\kappa)\varphi^1_\omega-(\kappa^2+1)\varphi^1_{\omega\omega}+2\varphi^2_\omega-(\varphi^1)^2-(\varphi^2)^2=0,\\
(\varphi^2-\kappa\varphi^1-2\kappa)\varphi^2_\omega-(\kappa^2+1)\varphi^2_{\omega\omega}-2\varphi^1_\omega=0.
\end{gather*}

\par\noindent
2.2. $\mathfrak g^{2.2}=\langle P^t,J\rangle$:\quad
$u=\dfrac{x\varphi^1-y\varphi^2}r+\dfrac x{r^2}$,\quad
$v=\dfrac{y\varphi^1+x\varphi^2}r+\dfrac y{r^2}$,\quad
where $\omega=r$;
\begin{gather*}
\varphi^1\varphi^1_\omega-\varphi^1_{\omega\omega}-\frac{(\varphi^2)^2}{\omega}-\frac1{\omega^3}=0,\\
\varphi^1\varphi^2_\omega-\varphi^2_{\omega\omega}+\frac{\varphi^1\varphi^2}{\omega}+2\frac{\varphi^2}{\omega^2}=0.
\end{gather*}

\par\noindent
2.3. $\mathfrak g^{2.3}=\langle D,J\rangle$:\quad
$u=\dfrac{x\varphi^1-y\varphi^2}{r\sqrt{|t|}}+\dfrac x{r^2}+\dfrac x{2t}$,\quad
$v=\dfrac{y\varphi^1+x\varphi^2}{r\sqrt{|t|}}+\dfrac y{r^2}+\dfrac y{2t}$,\quad
where \
$\omega=\dfrac r{\sqrt{|t|}}$;
\begin{gather*}
\varphi^1\varphi^1_\omega-\varphi^1_{\omega\omega}-\frac{(\varphi^2)^2}{\omega}-\frac1{\omega^3}-\frac\omega4=0,\\
\varphi^1\varphi^2_\omega-\varphi^2_{\omega\omega}+\frac{\varphi^1\varphi^2}{\omega}+2\frac{\varphi^2}{\omega^2}=0.
\end{gather*}

\medskip\par\noindent
2.4. $\mathfrak g^{2.4}=\langle P^t+\Pi,J\rangle$:\\[.5ex]
$u=\dfrac{x\varphi^1-y\varphi^2}{r\sqrt{t^2+1}}+\dfrac x{r^2}+\dfrac{tx}{t^2+1}$,\quad
$v=\dfrac{y\varphi^1+x\varphi^2}{r\sqrt{t^2+1}}+\dfrac y{r^2}+\dfrac{ty}{t^2+1}$,\quad
where \
$\omega=\dfrac r{\sqrt{t^2+1}}$;
\begin{gather*}
\varphi^1\varphi^1_\omega-\varphi^1_{\omega\omega}-\frac{(\varphi^2)^2}{\omega}-\frac1{\omega^3}+\omega=0,\\
\varphi^1\varphi^2_\omega-\varphi^2_{\omega\omega}+\frac{\varphi^1\varphi^2}{\omega}+2\frac{\varphi^2}{\omega^2}=0.
\end{gather*}

\medskip\par\noindent
2.5. $\mathfrak g^{2.5}_\mu=\langle P^t+\Pi+J+\mu(G^y+P^x),G^x-P^y\rangle_{\mu\geqslant0}$:\\[.5ex]
$u=\dfrac{\varphi^1-t\varphi^2+tx-y+\mu}{t^2+1}$,\quad
$v=\dfrac{t\varphi^1+\varphi^2+x+ty+\mu t}{t^2+1}$,\quad
where\quad
$\omega=\dfrac{x+ty}{t^2+1}-\mu\arctan t$;
\begin{gather*}
\varphi^1\varphi^1_\omega-\varphi^1_{\omega\omega}-2\varphi^2=0,\\
\varphi^1\varphi^2_\omega-\varphi^2_{\omega\omega}+2\varphi^1+2\mu=0.
\end{gather*}

\noprint{% Barannyks' version
\medskip\par\noindent
2.6. $\mathfrak g^{2.6}_{\kappa}=\langle G^x-P^y,G^y+P^x+\kappa P^y\rangle_{\kappa\geqslant0}$:\\[.5ex]
$u=\dfrac{(t+\kappa)\varphi^1-\varphi^2+(t+\kappa)x-y}{t(t+\kappa)+1}$,\quad
$v=\dfrac{\varphi^1+t\varphi^2+x+ty}{t(t+\kappa)+1}$,\quad
where\quad
$\omega=t$;
\quad$\varphi^1_\omega=\varphi^2_\omega=0$.
}

\medskip\par\noindent
2.6. $\mathfrak g^{2.6}_{\mu}=\langle G^x-P^y,G^y+\mu P^x\rangle_{\mu>0}$:\\[.5ex]
$u=\dfrac{t\varphi^1-\mu\varphi^2+tx-\mu y}{t^2+\mu}$,\quad
$v=\dfrac{\varphi^1+t\varphi^2+x+ty}{t^2+\mu}$,\quad
where\quad
$\omega=t$;
%\quad$\varphi^1_\omega=\varphi^2_\omega=0$.
\begin{gather*}
\varphi^1_\omega=\varphi^2_\omega=0.
\end{gather*}

%\bigskip\par
The differential constraint $u_y=v_x$ singling out the widest linearizable solution subset of the system~\eqref{eq:BurgersSystem}
is respectively reduced by ansatzes~2.1--2.6 to the following differential constraints
in terms of invariant variables:
\begin{gather*}
2.1.\ \varphi^1_\omega+\kappa\varphi^2_\omega=0.\quad
\mbox{2.2.--2.4.}\ \omega\varphi^2_\omega+\varphi^2=0.\quad
2.5.\ \varphi^2_\omega=-2.\quad
2.6.\ 1=0.
\end{gather*}
For each of the above reduced systems of ODEs,
only solutions that do not satisfy the differential constraint with the same number
may give new exact solutions of the system~\eqref{eq:BurgersSystem}.

The maximal Lie invariance algebras of reduced systems~2.1--2.6 are the following:
\begin{gather*}
2.1.\ \langle\p_\omega\rangle; \quad
2.2.\ \langle\omega\p_\omega-\varphi^1\p_{\varphi^1}-\varphi^2\p_{\varphi^2}\rangle; \quad
2.3.\ \{0\};\quad
2.4.\ \{0\};\quad
2.5.\ \langle\p_\omega\rangle;\\
2.6.\ \langle\xi(\omega,\varphi^1,\varphi^2)\p_\omega+\eta^1(\varphi^1,\varphi^2)\p_{\varphi^1}+\eta^2(\varphi^1,\varphi^2)\p_{\varphi^2}\rangle,
\end{gather*}
where $\xi$, $\eta^1$ and~$\eta^2$ run through the sets of smooth functions of their arguments.
Analogously to Lie reductions of codimension one,
we use the procedure of determining which Lie symmetries of reduced systems~2.1--2.6 are induced
by Lie symmetries of the original Burgers system~\eqref{eq:BurgersSystem}.
For each $m\in\{1,\dots,6\}$ we find the normalizer of the subalgebra~$\mathfrak g^{2.m}$
in the algebra~$\mathfrak g$,
${\rm N}_{\mathfrak g}(\mathfrak g^{2.m})=\{Q\in\mathfrak g\mid [Q,Q']\in\mathfrak g^{2.m}\ \mbox{for all}\ Q'\in\mathfrak g^{2.m}\}$.
We get
\begin{gather*}
{\rm N}_{\mathfrak g}(\mathfrak g^{2.1}_\kappa)=
{\rm N}_{\mathfrak g}(\mathfrak g^{2.2})=\langle P^t, D,J\rangle,\quad
{\rm N}_{\mathfrak g}(\mathfrak g^{2.3})=\langle D,J\rangle,\quad
{\rm N}_{\mathfrak g}(\mathfrak g^{2.4})=\langle P^t+\Pi,J\rangle,\\
{\rm N}_{\mathfrak g}(\mathfrak g^{2.5}_\mu)=\langle P^t+\Pi+J,G^y+P^x,G^x-P^y\rangle,\\
{\rm N}_{\mathfrak g}(\mathfrak g^{2.6}_\kappa)=\langle P^x,P^y,G^x,G^y,\mu P^t+\Pi,(\mu+1)P^t-J\rangle.
\end{gather*}
The algebra of induced Lie symmetries of reduced system~2.$m$
is isomorphic to the quotient algebra ${\rm N}_{\mathfrak g}(\mathfrak g^{2.m})/\mathfrak g^{2.m}$.
Therefore, all Lie symmetries of reduced system~2.$m$
are induced by Lie symmetries of the original Burgers system~\eqref{eq:BurgersSystem}
if and only if  $\dim\mathfrak a^m=\dim{\rm N}_{\mathfrak g}(\mathfrak g^{2.m})-2$.
This is the case for each $m\in\{1,\dots,5\}$.
In other words, the Burgers system~\eqref{eq:BurgersSystem}
possesses no hidden symmetries associated with reductions~2.1--2.5.
Reduction~2.6 is singular since the corresponding reduced system consists of two first-order ODEs.
The maximal Lie invariance algebra of this system is infinite-dimensional.
Therefore, the system~\eqref{eq:BurgersSystem} admits many hidden symmetries
related to reduction~2.6
although these symmetries are not of interest
since reduced system~2.6 is trivially integrable, $\varphi^1,\varphi^2=\const$, i.e., $\varphi^1,\varphi^2=0\bmod G$,
and the associated solutions of the system~\eqref{eq:BurgersSystem} are quite simple.

\section{Stationary similarity solutions}\label{sec:SolutionOfReducedSystemsOfODEs}

Following~\cite[Eq.~(58)]{raja2008a}, we take another ansatz for the algebra~$\mathfrak g^{2.1}_0=\langle P^t,D\rangle$,
%$u=y^{-1}\varphi^1$, $v=y^{-1}\varphi^2$, $\omega=x/y$.
\begin{gather}\label{eq:BurgersSystemAnzatz2.1kappa0}
u=\frac{\varphi^1(\omega)}y,\quad
v=\frac{\varphi^2(\omega)}y\quad\mbox{with}\quad
\omega=\frac xy.
\end{gather}

The corresponding reduced system is
\begin{gather}\label{eq:BurgersSystemNewReducedSystem2.1kappa0}
(1+\omega^2)\varphi^i_{\omega\omega}+4\omega\varphi^i_\omega+2\varphi^i+\varphi^2(\omega\varphi^i_\omega+\varphi^i)-\varphi^1\varphi^i_\omega=0,\quad i=1,2.
\end{gather}
In order to construct particular solutions of the system~\eqref{eq:BurgersSystemNewReducedSystem2.1kappa0},
we set the additional constraint $\varphi^2_\omega=0$,
under which the second equation of the system~\eqref{eq:BurgersSystemNewReducedSystem2.1kappa0} implies $\varphi^2\in\{0,-2\}$.
Substituting each of these values of~$\varphi^2$ into the first equation of~\eqref{eq:BurgersSystemNewReducedSystem2.1kappa0},
we obtain a completely integrable equation.
As the first step of integration, we integrate this equation once to the Riccati equation
\begin{equation}\label{eq:BurgersSystemRicattiEq}
(1+\omega^2)\varphi^1_\omega+(\varphi^2+2)\omega\varphi^1-\frac12(\varphi^1)^2=-2A, \quad \mbox{where}\quad A=\const,
\end{equation}
which is reduced, in the standard way, by the substitution
\begin{equation}\label{eq:BurgersSystemRicattiSubs}
\varphi^1=-2(1+\omega^2)\frac{\psi_\omega}\psi\quad\mbox{with}\quad \psi=\psi(\omega)
\end{equation}
into the second-order linear ODE
\begin{equation}\label{eq:BurgersSystem2ndOrderLinEq}
(1+\omega^2)^2\psi_{\omega\omega}+(\varphi^2+4)\omega(1+\omega^2)\psi_\omega-A\psi=0.
\end{equation}

For $\varphi^2=0$, the equation~\eqref{eq:BurgersSystem2ndOrderLinEq} was integrated in elementary functions in~\cite{raja2008a}.
The form of the general solution depends on the constant parameter~$A$,
\begin{gather*}
\psi=C_1+C_2\left(\arctan\omega+\frac\omega{1+\omega^2}\right)\quad\mbox{if}\quad A=0,
\\[.5ex]
\psi=C_1\frac\omega{\sqrt{1+\omega^2}}+C_2\frac{\omega\arctan\omega+1}{\sqrt{1+\omega^2}}\quad\mbox{if}\quad A=1,
\\[1ex]
\psi=
C_1\frac{\omega-\alpha}{\sqrt{1+\omega^2}}e^{-\alpha\arctan\omega}+
C_2\frac{\omega+\alpha}{\sqrt{1+\omega^2}}e^{\alpha\arctan\omega}\quad\mbox{with}\quad \alpha=\sqrt{A-1} \quad\mbox{if}\quad A>1,
\\[1ex]
\psi=
C_1\frac{\omega\cos(\beta\arctan\omega)-\beta\sin(\beta\arctan\omega)}{\sqrt{1+\omega^2}}+
C_2\frac{\omega\sin(\beta\arctan\omega)+\beta\cos(\beta\arctan\omega)}{\sqrt{1+\omega^2}}\\
\qquad\mbox{with}\quad \beta=\sqrt{1-A} \quad\mbox{if}\quad A<1\quad\mbox{and}\quad A\ne0.
\end{gather*}
Here and in what follows $C$'s are arbitrary constants.
For integer values of~$\beta$, the last solution can be rewritten in terms of the polynomials
\[
P(\omega)=\sum_{i=0}^{\bigl[\frac\beta2\bigr]}(-1)^i\binom{\beta}{2i}\omega^{\beta-2i}, \quad
Q(\omega)=\sum_{i=0}^{\bigl[\frac{\beta+1}2\bigr]+1}(-1)^i\binom{\beta}{2i+1}\omega^{\beta-2i-1},
\]
\[
\psi=
C_1\frac{\omega P(\omega)-\beta Q(\omega)}{(1+\omega^2)^{\frac{\beta+1}2}}+
C_2\frac{\omega Q(\omega)+\beta P(\omega)}{(1+\omega^2)^{\frac{\beta+1}2}}.
\]

The integration of the equation~\eqref{eq:BurgersSystem2ndOrderLinEq} with $\varphi^2=-2$ is much simpler.
This equation can be represented in the form $(1+\omega^2)((1+\omega^2)\psi_\omega)_\omega-A\psi=0$
and hence it is reduced to the constant-coefficient linear ODE $\psi_{\zeta\zeta}=A\psi$
by the change of the independent variable $\zeta=\arctan\omega$.
Depending on the constant parameter~$A$, we obtain the general solution
\begin{gather*}
\psi=C_1+C_2\arctan\omega\quad\mbox{if}\quad A=0,
\\[1ex]
\psi=C_1e^{-\alpha\arctan\omega}+C_2e^{\alpha\arctan\omega}\quad\mbox{with}\quad \alpha=\sqrt A \quad\mbox{if}\quad A>0,
\\[1ex]
\psi=C_1\cos(\beta\arctan\omega)+C_2\sin(\beta\arctan\omega)\quad\mbox{with}\quad \beta=\sqrt{-A} \quad\mbox{if}\quad A<0.
\end{gather*}
Again, for integer values of~$\beta$, the last solution can be rewritten in terms of the above polynomials~$P$ and~$Q$,
\[
\psi=\frac{C_1P(\omega)+C_2Q(\omega)}{(1+\omega^2)^{\beta/2}}.
\]

Obtained expressions for the function~$\psi$, depending on either $\varphi^2=0$ or $\varphi^2=-2$,
together with ansatz~\eqref{eq:BurgersSystemAnzatz2.1kappa0} and
representation~\eqref{eq:BurgersSystemRicattiSubs} for the function~$\varphi^1$
give stationary similarity solutions of the Burgers system~\eqref{eq:BurgersSystem}.

The ansatz~\eqref{eq:BurgersSystemAnzatz2.1kappa0} reduces the differential constraint~\eqref{eq:BurgersSystemLinearizationConstraint} %, $u_y=v_x$,
to the equation $\omega\varphi^1_\omega+\varphi^1+\varphi^2_\omega=0$.
This equation is satisfied by solutions of the system~\eqref{eq:BurgersSystemNewReducedSystem2.1kappa0} with $\varphi^2=\const$
only if $\varphi^1=0$ or $\varphi^1=-2/\omega$,
which is a very minor subset of the set of above solutions.
Moreover, constructed solutions of the Burgers system~\eqref{eq:BurgersSystem} are equivalent to shift-invariant ones
only if $\varphi^1=-2/\omega$ for $\varphi^2=0$ or $\varphi^1=\const$.

\noprint{
%%% The condition $\varphi^1=\const$ is induced by the differential constraint $u_y=v_x$ !!!
Ansatz~2.1 and the corresponding reduced system can also be used
for the construction of other $\mathfrak g^{2.1}_0$-invariant solutions of the system~\eqref{eq:BurgersSystem}.
Suppose that $\varphi^1=\const$.
Then from reduced system~2.1 we have two equations for $\varphi^2$,
$2\varphi^2_\omega=(\varphi^2)^2+(\varphi^1)^2$ and
$\varphi^2_{\omega\omega}=\varphi^2\varphi^2_\omega$,
and the second equation is a differential consequence of the first one,
which is a constant coefficient Riccati equation.
The general solution of this equation is
\[
\varphi^2=2C_0\tan(C_0\omega+C_1) \quad\mbox{if}\quad C_0:=\frac{\varphi^1}2\ne0
\quad\mbox{and}\quad
\varphi^2=-\frac1{\omega+C_1} \quad\mbox{if}\quad \varphi^1=0.
\]
Using rotation Lie symmetries generated by~$J$, we can set the constant~$C_1$ to be equal 0.
As a result, we construct the following solutions of the system~\eqref{eq:BurgersSystem}:
\begin{gather*}
u=\frac{2C_0}{r^2}\left(x-y\tan\left(C_0\arctan\frac yx\right)\right), \quad
v=\frac{2C_0}{r^2}\left(y+x\tan\left(C_0\arctan\frac yx\right)\right),
\\
u=\frac y{r^2\arctan(y/x)},\quad
v=-\frac x{r^2\arctan(y/x)}.
\end{gather*}
}

\section{Conservation laws}\label{sec:BurgersSystemCLs}

In contrast to Lie symmetries, properties of the Burgers system~\eqref{eq:BurgersSystem} related to local conservation laws
are poor.
(See \cite{boch1999a,olve1993b,popo2008a,vino1984a} for definitions of related notions and necessary theoretical results.)

\begin{proposition}\label{pro:BurgersSystemCLs}
The Burgers system~\eqref{eq:BurgersSystem} admits no nontrivial cosymmetries
and, therefore, no nonzero local conservation laws.
\end{proposition}

\begin{proof}
It is convenient to re-denote the independent and dependent variables as
$z_0=t$, $z_1=x$, $z_2=y$, $w^1=u$, $w^2=v$
and thus $z=(z_0,z_1,z_2)$ and $w=(w^1,w^2)$ are the tuples of independent and dependent variables, respectively.
A differential function $F=F[w]$ of~$w$ is, roughly speaking, a smooth function of~$z$, $w$
and a finite number of derivatives of~$w$ with respect to~$z$.
Cosymmetries of the system~\eqref{eq:BurgersSystem} are pairs of differential functions of $w$,
$\gamma=(\gamma^1[w],\gamma^2[w])$ that satisfy the system
\begin{gather}\label{eq:BurgersSystemConditionFoCosyms}
-\mathrm D_0\gamma^i-\mathrm D_j(w^j\gamma^i)-\mathrm D_j\mathrm D_j\gamma^i+w^j_{\delta_i}\gamma^j=0
\end{gather}
on the manifold defined by the system~\eqref{eq:BurgersSystem} and its differential consequences
in the corresponding jet space $J^{\infty}(z|w)$.
Here $\mathrm D_\mu$ denotes the operator of total derivative with respect to the variable~$z_\mu$.
In other words, $\mathrm D_\mu=\partial_\mu+w^i_{\alpha+\delta_\mu}\partial_{w^i_\alpha}$,
where $\alpha=(\alpha_0,\alpha_1,\alpha_2)$ is an arbitrary multi-index, $\alpha_\mu\in\mathbb N_0=\mathbb N\cup\{0\}$.
The indices~$i$ and~$j$ run from 1 to~2, the index~$\mu$ runs from 0 to~2,
and the summation convention over repeated indices is used.
The variable~$w^i_\alpha$ of the jet space $J^{\infty}(z|w)$ is identified with the derivative of $w^i$ of order~$\alpha$,
$w^i_\alpha=\partial^{|\alpha|}w^i/\partial z_0^{\alpha_0}\partial z_1^{\alpha_1}\partial z_2^{\alpha_2}$,
$|\alpha|=\alpha_0+\alpha_1+\alpha_2$,
and $\delta_\mu$ is the multi-index with zeros everywhere except on the $\mu$th entry, which equals $1$.
The order of a differential function~$F$ is equal to the highest order of jet variables involved in~$F$,
where $\ord z_\mu=-\infty$ and $\ord w^i_\alpha=|\alpha|$.

Since the system~\eqref{eq:BurgersSystem} is of evolution type,
the cosymmetry~$\gamma$ can be assumed, up to equivalence of cosymmetries,
not to depend on derivatives of~$w$ involving differentiation with respect to~$z_0$.

Suppose that $\ord\gamma>-\infty$.
Up to permutation of~$(z_1,w^1)$ and $(z_2,w^2)$, which is a Lie symmetry transformation for the system~\eqref{eq:BurgersSystem},
we can assume that $\ord\gamma^1\geqslant\ord\gamma^2$ and thus $\ord\gamma^1>-\infty$.
After expanding the equation~\eqref{eq:BurgersSystemConditionFoCosyms} with $i=1$ and
substituting the expressions implied by the system~\eqref{eq:BurgersSystem}
and its differential consequences for derivatives $w^i_\alpha$ with $\alpha_0=1$,
we collect the terms with derivatives of~$w$ of the highest order $r:=\ord\gamma^1+2$ appearing in this equation,
which gives \[-2\sum_{|\alpha|=r}\gamma^1_{w^i_\alpha}(w^i_{\alpha+2\delta_1}+w^i_{\alpha+2\delta_2})=0.\]
Splitting this equality with respect to $(r+2)$th order derivatives of~$w$
implies that $\gamma^1_{w^i_\alpha}=0$ for any $i=1,2$ and any~$\alpha$ with $|\alpha|=\ord\gamma^1$,
which contradicts the definition of order of a differential function.
(The above consideration is similar to the derivation of upper bound
for orders of conservation laws of (1+1)-dimensional even-order evolution equations~\cite[Section~22.5]{Ibragimov1985}.
See also analogous results for systems of such equations~\cite{folt1999a} and more general systems~\cite{igon2002b}.)

Therefore, $\ord\gamma=-\infty$.
Collecting coefficients of $w^2_1$ and $w^1_2$ in the equations~\eqref{eq:BurgersSystemConditionFoCosyms}
with $i=1$ and $i=2$, respectively, we derive $\gamma^1=\gamma^2=0$.
In other words, the Burgers system~\eqref{eq:BurgersSystem} admits only trivial cosymmetries.
Since each conservation-law characteristic of a system of differential equations is a cosymmetry of this system
and for systems in the extended Kovalevskaya form, which include all systems of evolution equations,
trivial characteristics are associated with trivial conserved currents~\cite{mart79a},
the space of conservation laws of the Burgers system~\eqref{eq:BurgersSystem} is zero-dimensional.
\end{proof}

In spite of absence of local conservation laws for the Burgers system~\eqref{eq:BurgersSystem},
we can consider local conservation laws of various submodels related to this model,
and such conservation laws can be interpreted as its hidden conservation laws.

Thus, for the (1+2)-dimensional linear heat equation~\eqref{eq:(1+2)DLinHeatEq}
the space of its reduced conservation-law characteristics is infinite dimensional and
can be identified with the solution space of the (1+2)-dimensional backward linear heat equation,
$\lambda_t+\lambda_{xx}+\lambda_{yy}=0$ with $\lambda=\lambda(t,x,y)$.
The relation of the equations~\eqref{eq:(1+2)DLinHeatEq} and~\eqref{eq:BurgersSystemEqLinearizableTo(1+2)DLinHeatEq}
via the transformation $\phi=e^{-\psi/2}$ allows us
to conclude that the analogous space for the equation~\eqref{eq:BurgersSystemEqLinearizableTo(1+2)DLinHeatEq}
can be identified with $\{\lambda(t,x,y)e^{-\psi/2}\}$,
where the function~$\lambda$ again runs through the solution space of the (1+2)-dimensional backward linear heat equation.
The last claim implies that the system~$\mathcal S$,
which consists of the Burgers system~\eqref{eq:BurgersSystem}
jointly with the differential constraint~\eqref{eq:BurgersSystemLinearizationConstraint},
admits the only independent conservation law related to the conserved form
of the equation~\eqref{eq:BurgersSystemLinearizationConstraint}.
Indeed, if the system~$\mathcal S$ had admitted another independent conservation law,
this conservation law would induce a conservation law of the equation~\eqref{eq:BurgersSystemEqLinearizableTo(1+2)DLinHeatEq}
whose characteristic does not depend on~$\psi$ and thus cannot be reduced to the form~$\lambda(t,x,y)e^{-\psi/2}$.

Reduced systems~1.6--1.8 are systems of even-order evolution quasilinear equations.
In a way similar to the proof of Proposition~\ref{pro:BurgersSystemCLs},
it is possible to show that these systems may possess only reduced cosymmetries of order $-\infty$.
In fact, all cosymmetries of reduced systems~1.6 and~1.7 are trivial
and hence the spaces of local conservation laws of these systems are zero-dimensional
whereas the similar space for reduced system~1.8 is one-dimensional,
and its space of reduced characteristics is spanned by the characteristic $(1,0)$.
The single linearly independent local conservation law of reduced system~1.8
can be interpreted as the pullback of the single linearly independent local conservation law
of the first equation of this system, which is no other than the classical Burgers equation.
Conservation laws of the Burgers equation are well known;
see \mbox{\cite[Chapter~5, Example~3.1]{boch1999a}} and~\cite{popo2005b}.
Using the conserved current $(w^1,(w^1)^2/2-w^1_2)$ associated with the characteristic $(1,0)$,
we introduce the potential~$\psi$ defined by the equations $\psi_2=w^1$ and $\psi_1=w^1_2-(w^1)^2/2$.
These equations jointly with the second equation of reduced system~1.8 constitute
a potential system for this reduced system.
We can exclude~$w^1$ from the potential system in view of the equation $w^1=\psi_x$.
There is a one-to-one correspondence between conservation laws of the complete potential system
and the system for $(w^2,\psi)$,
which is established by the pullback with the projection $(w^1,w^2,\psi)\to(w^2,\psi)$.
The system for $(w^2,\psi)$,
\[
\psi_1+\frac{(\psi_2)^2}2-\psi_{22}=0,\quad
w^2_1+\psi_2w^2_2-w^2_{22}=0,
\]
is linearized by the point transformation of dependent variables
$\psi=-2\ln\tilde w^1$, $w^2=\tilde w^2/\tilde w^1$
to the decoupled system~\eqref{eq:SystemOfTwoLinHeatEqs} of two copies of the linear heat equation.
In this way, we represent the process of linearizing reduced system~1.8
as a composition of potentialization, projection and point transformation.
Then each local conservation of the system~\eqref{eq:SystemOfTwoLinHeatEqs}
can be assumed as a potential conservation law of reduced system~1.8.
There are no more potential conservation laws of reduced system~1.8
\cite{popo2005b,popo2008a}.
Therefore, its space of potential conservation laws is parameterized
by pairs of solutions of the (1+1)-dimensional backward linear heat equation.

It is easy to prove that reduced systems~1.1--1.5 do not possess local conservation laws
with characteristics of order not greater than zero.
We can conjecture that these systems admit no local conservation laws at all
but the proof of this conjecture is expected to be cumbersome.

\section{Common solutions of Burgers system\\ and Navier--Stokes equations}\label{sec:BurgersSystemSolutionsCommonWithNavier-StokesEqns}

In spite of a formal similarity in the form of certain equations,
the Burgers system~\eqref{eq:BurgersSystem} and
the (1+2)-dimensional incompressible Navier--Stokes equations,
\begin{gather}\label{eq:NSEqs}
\begin{split}
&u_t+uu_x+vu_y-u_{xx}-u_{yy}+p_x=0,\\
&v_t+uv_x+vv_y-v_{xx}-v_{yy}+p_y=0,\\
&u_x+v_y=0,
\end{split}
\end{gather}
are in fact not close to each other.
In~\eqref{eq:NSEqs}, $(u,v)$ is the flow velocity, $p$ is the pressure,
the kinematic viscosity and the fluid density are set, without loss of generality,
to be equal 1.
If we attach the differential constraint $u_x+v_y=0$ to the system~\eqref{eq:BurgersSystem},
then solutions of the joint system
\begin{subequations}\label{eq:BurgersSystemWithIncompressibilityCondition}
\begin{gather}
R^1:=u_t+uu_x+vu_y-u_{xx}-u_{yy}=0,\label{eq:BurgersSystemWithIncompressibilityConditionA}\\
R^2:=v_t+uv_x+vv_y-v_{xx}-v_{yy}=0,\label{eq:BurgersSystemWithIncompressibilityConditionB}\\
R^3:=u_x+v_y=0\label{eq:BurgersSystemWithIncompressibilityConditionC}
\end{gather}
\end{subequations}
are prolonged with $p=\const$ to solutions of the system~\eqref{eq:NSEqs}
and thus they can be interpreted as common solutions of the Burgers system~\eqref{eq:BurgersSystem}
and the Navier--Stokes equations~\eqref{eq:NSEqs}.

To describe the solution set of the overdetermined system~\eqref{eq:BurgersSystemWithIncompressibilityCondition},
we compute various differential consequences of this system. In particular,
\begin{gather*}
\mathrm D_xR^1+\mathrm D_yR^2-(\mathrm D_t+u\mathrm D_x+v\mathrm D_y-\mathrm D_x^2-\mathrm D_y^2)R^3-(R^3)^2
=-2(u_xv_y-u_yv_x)=0,
\end{gather*}
where $\mathrm D_t$, $\mathrm D_x$ and $\mathrm D_y$ denote the operators of total derivatives 
with respect to the variables $t$, $x$ and~$y$, respectively;
see Section~\ref{sec:BurgersSystemCLs}.

In other words, the Jacobian of $(t,u,v)$ with respect to $(t,x,y)$ vanishes
and thus the independent variable~$t$ and the unknown functions~$u$ and~$v$ are functionally dependent.
Up to permutation of $(x,u)$ and $(y,v)$,
we can assume that $v$ is represented as a function of~$t$ and~$u$, $v=F(t,u)$.
This gives  $R^3=u_x+F_uu_y=0$, i.e., $u_x=-F_uu_y$, and thus
\[
R^2-F_uR^1=F_t-F_{uu}(u_x{}^{\!2}+u_y{}^{\!2})=F_t-F_{uu}((F_u)^2+1)u_y{}^{\!2}=0.
\]

Suppose that $F_{uu}\ne0$.
Then the last equation can be solved with respect to~$u_y$,
which gives the representation of~$u_y$ as a function of~$(t,u)$, $u_y=G(t,u)$.
The cross-differentiation of the equations $u_y=G$ and $u_x=-F_uG$ with respect to~$x$ and~$y$
leads to the constraint $F_{uu}G^2=0$, which is equivalent to $G=0$ since $F_{uu}\ne0$.
As a result, $u_x=u_y=v_x=v_y=0$
and the equations~\eqref{eq:BurgersSystemWithIncompressibilityConditionA}
and~\eqref{eq:BurgersSystemWithIncompressibilityConditionB} also imply $u_t=v_t=0$.
Therefore, solutions with $F_{uu}\ne0$ are exhausted by trivial constant solutions.

If  $F_{uu}=0$, then also $F_t=0$.
This means that $v=c_1u+c_0$ and $u_x+c_1u_y=0$ for some constants $c_0$ and~$c_1$.

Uniting and symmetrizing the above solutions, we obtain the following assertion.

\begin{proposition}\label{pro:BurgersSystemWithIncompressibilityConditionSolutionSet}
The tuple $(u,v)$ satisfies the system~\eqref{eq:BurgersSystemWithIncompressibilityCondition}
if and only if its components are of the form
$u=a_1w(t,z)+a_0$, $v=a_2w(t,z)$,
where $a_0$, $a_1$ and~$a_2$ are constants with $a_1{}^{\!2}+a_2{}^{\!2}=1$,
$z=a_2x-a_1y$, and $w=w(t,z)$ is an arbitrary solution of the linear equation $w_t+a_0w_z=w_{zz}$.
\end{proposition}

Using symmetry transformations of rotations and Galilean boosts,
we can set $a_0=a_1=0$ and $a_2=1$,
which leads to the particular subcase of Lie reduction~1.8 of the system~\eqref{eq:BurgersSystem} with $w^1=0$.

\section{One more reduction to single (1+2)-dimensional PDE}
\label{sec:OneMoreReductionToSingle(1+2)DPDE}

An interesting subset of solutions of the system~\eqref{eq:BurgersSystem} is singled out by the differential constraint
\begin{gather}\label{eq:BurgersSystemDiffConstraintForSingle(1+2)DPDE}
u_x=v_y.
\end{gather}
Although this constraint differs from the ``incompressibility'' constraint~\eqref{eq:BurgersSystemWithIncompressibilityConditionC}
only by the sign relating~$u_x$ and~$v_y$,
the solution set of the joint system~\eqref{eq:BurgersSystem},~\eqref{eq:BurgersSystemDiffConstraintForSingle(1+2)DPDE},
which we denote by~$\bar{\mathcal S}$,
is essentially different from that the system~\eqref{eq:BurgersSystemWithIncompressibilityCondition}.
In contrast to the formally similar constraints~\eqref{eq:BurgersSystemLinearizationConstraint}
and~\eqref{eq:BurgersSystemWithIncompressibilityConditionC}, $u_y=v_x$ and $u_x+v_y=0$,
the constraint~\eqref{eq:BurgersSystemDiffConstraintForSingle(1+2)DPDE} is not rotationally invariant
and hence it is a representative of a family of rotationally equivalent constraints.

At the same time, the constraint~\eqref{eq:BurgersSystemDiffConstraintForSingle(1+2)DPDE}
is similar to the constraint~\eqref{eq:BurgersSystemLinearizationConstraint}
in the sense that it allows one to reduce the system~\eqref{eq:BurgersSystem}
to a single (1+2)-dimensional PDE via introducing a potential
although, in contrast to~\eqref{eq:BurgersSystemEqLinearizableTo(1+2)DLinHeatEq},
this equation is not linearizable.
Using the equation~\eqref{eq:BurgersSystemDiffConstraintForSingle(1+2)DPDE}, which is of conserved form,
as a ``short'' conservation law, we introduce
the potential~$\psi$ defined by the equations $\psi_x=v$ and $\psi_y=u$.
The substitution of the expression of $(u,v)$ in terms of~$\psi$ reduces the system~\eqref{eq:BurgersSystem}
to the condition that the derivatives of $R=\psi_t+\psi_x\psi_y-\psi_{xx}-\psi_{yy}$
with respect to~$y$ and~$x$ vanish, i.e., the function~$R$ depends only on~$t$.
Since the potential~$\psi$ is defined up to summand being an arbitrary smooth function of~$t$,
we can set the function~$R$ vanish,
which gives the equation for the potential~$\psi$,
\begin{gather}\label{eq:BurgersSystemSingle(1+2)DPDE}
\psi_t+\psi_x\psi_y-\psi_{xx}-\psi_{yy}=0.
\end{gather}
Each known solution of the equation~\eqref{eq:BurgersSystemSingle(1+2)DPDE}
leads to a solution of the Burgers system~\eqref{eq:BurgersSystem}
\begin{gather}\label{eq:BurgersSystemSolutionsViaSingle(1+2)DPDE}
u=\psi_y, \quad v=\psi_x.
\end{gather}

\subsection{Symmetries of overdetermined system}

\begin{proposition}\label{pro:BurgersSystemWithDiffConstraintForSingle(1+2)DPDELieSyms}
The maximal Lie invariance algebra~$\bar{\mathfrak g}$ of the system~$\bar{\mathcal S}$ coincides with
\[\mathfrak g'=\langle P^t,D,\Pi,P^x,P^y,G^x,G^y\rangle\simeq\mathfrak g/\langle J\rangle,\]
where $\mathfrak g$ is the maximal Lie invariance algebra of the system~\eqref{eq:BurgersSystem}.
The complete point symmetry group~$\bar G$ of the system~$\bar{\mathcal S}$
is generated by one-parameter groups associated with vector fields from the algebra~$\mathfrak g'$
and two discrete transformations: the transformation of simultaneous mirror mappings in the $(x,y)$- and $(u,v)$-planes,
$I_1\colon(t,x,y,u,v)\mapsto(t,-x,y,-u,v)$, and the permutation of~$(x,u)$ and~$(y,v)$, $I_0\colon(t,x,y,u,v)\mapsto(t,y,x,v,u)$.
\end{proposition}

Therefore, the group~$\bar G$ consists of the transformations of the form 
\begin{gather}\label{txyComponentsOfPOintSymsOfTheArtifitialEq1}
\tilde t=\frac{\alpha t+\beta}{\gamma t+\delta},
\quad
\begin{pmatrix}\tilde x\\ \tilde y\end{pmatrix}
=\frac\sigma{\gamma t+\delta}\begin{pmatrix}\varepsilon_1x\\ \varepsilon_2y\end{pmatrix}
+\frac{\alpha t+\beta}{\gamma t+\delta}\begin{pmatrix}\mu_1\\ \mu_2\end{pmatrix}
+\begin{pmatrix}\nu_1\\ \nu_2\end{pmatrix},
\\[1ex]\nonumber
\begin{pmatrix}\tilde u\\ \tilde v\end{pmatrix}
=\frac{\gamma t+\delta}{\sigma}\begin{pmatrix}\varepsilon_1u\\ \varepsilon_2v\end{pmatrix}
-\frac\gamma\sigma\begin{pmatrix}\varepsilon_1x\\ \varepsilon_2y\end{pmatrix}
+\begin{pmatrix}\mu_1\\ \mu_2\end{pmatrix}
\end{gather}
or
\begin{gather}\label{txyComponentsOfPOintSymsOfTheArtifitialEq2}
\tilde t=\frac{\alpha t+\beta}{\gamma t+\delta},
\quad
\begin{pmatrix}\tilde x\\ \tilde y\end{pmatrix}
=\frac\sigma{\gamma t+\delta}\begin{pmatrix}\varepsilon_2y\\ \varepsilon_1x\end{pmatrix}
+\frac{\alpha t+\beta}{\gamma t+\delta}\begin{pmatrix}\mu_1\\ \mu_2\end{pmatrix}
+\begin{pmatrix}\nu_1\\ \nu_2\end{pmatrix},
\\[1ex]\nonumber
\begin{pmatrix}\tilde u\\ \tilde v\end{pmatrix}
=\frac{\gamma t+\delta}{\sigma}\begin{pmatrix}\varepsilon_2v\\ \varepsilon_1u\end{pmatrix}
-\frac\gamma\sigma\begin{pmatrix}\varepsilon_2y\\ \varepsilon_1x\end{pmatrix}
+\begin{pmatrix}\mu_1\\ \mu_2\end{pmatrix},
\end{gather}
where
$\alpha$, $\beta$, $\gamma$ and $\delta$ are arbitrary constants with $\alpha\delta-\beta\gamma>0$ 
such that their tuple is defined up to nonvanishing multiplier, $\sigma=\sqrt{\alpha\delta-\beta\gamma}$,
$\varepsilon_1,\varepsilon_2=\pm1$,
and $\mu_1$, $\mu_2$, $\nu_1$ and $\nu_2$ are arbitrary constants. 

\begin{proof}
Since the system~$\bar{\mathcal S}$ is alike to the system~$\mathcal S$ in overdetermination degree,
for computing its maximal Lie invariance algebra
we follow the proof of Proposition~\ref{pro:BurgersSystemWithLinearizationConstraintLieSyms}.
The description of the solution set of~$\bar{\mathcal S}$ in terms of
the representation $(u,v)=(\psi_y,\psi_x)$ 
and the solution set of the equation~\eqref{eq:BurgersSystemSingle(1+2)DPDE} implies
that independent differential consequences of~$\bar{\mathcal S}$
whose orders as differential equations equal one or two
are exhausted by the equation~\eqref{eq:BurgersSystemDiffConstraintForSingle(1+2)DPDE}
and the equations~\eqref{eq:BurgersSystem}, $u_{xx}=v_{xy}$ and $u_{xy}=v_{yy}$, respectively.
It is obvious that the last two equations are differential consequences of~\eqref{eq:BurgersSystemDiffConstraintForSingle(1+2)DPDE}
that are obtained by single differentiations with respect to~$x$ and~$y$.
This is why we separately apply the infinitesimal invariance criterion
to the equation~\eqref{eq:BurgersSystemDiffConstraintForSingle(1+2)DPDE}
and then to the equations~\eqref{eq:BurgersSystem},
substitute for derivatives in view of all the above differential consequences
and split with respect to parametric derivatives.
Solving the constructed system of determining equations leads to the algebra~$\mathfrak g'$.

The complete point symmetry group~$\bar G$ of the system~$\bar{\mathcal S}$
is efficiently computed by the suggested version of the algebraic method.
Similarly to the proof of Theorem~\ref{thm:BurgersSystemCompletePointSymGroup},
in order to find discrete point symmetries of~$\bar{\mathcal S}$
it suffices to consider automorphisms of~$\mathfrak g'$ whose matrices are of the form
$\mathop{\rm diag}(\varepsilon,1,\varepsilon)\oplus\tilde A$,
where $\varepsilon=\pm1$ and $\tilde A$ is a $4\times4$ nondegenerate matrix.
The set of such automorphisms is exhausted by those with matrices of the form
\[
A=\mathop{\rm diag}(\varepsilon,1,\varepsilon)
\oplus\begin{pmatrix}a&b\\c&d\end{pmatrix}
\oplus\varepsilon\begin{pmatrix}a&b\\c&d\end{pmatrix},
\]
where $\varepsilon=\pm1$ and $ad-bc\ne0$.
Suppose that the pushforward~$\mathcal T_*$ of vector fields by the point transformation
$
\mathcal T\colon (\tilde t,\tilde x,\tilde y,\tilde u,\tilde v)=(T,X,Y,U,V)(t,x,y,u,v)
$
is the automorphism of~$\mathfrak g'$ with the matrix~$A$, i.e.,
\begin{gather*}
\mathcal T_*P^t=\varepsilon\tilde P^t,\quad
\mathcal T_*P^x=a\tilde P^x+c\tilde P^y,\quad
\mathcal T_*P^y=b\tilde P^x+d\tilde P^y,\\
\mathcal T_*G^x=\varepsilon a\tilde G^x+\varepsilon c\tilde G^y,\quad
\mathcal T_*G^y=\varepsilon b\tilde G^x+\varepsilon d\tilde G^y,\quad
\mathcal T_*D  =\tilde D,\quad
\mathcal T_*\Pi=\varepsilon\tilde\Pi,
\end{gather*}
where tildes over vector fields mean that these vector fields are given in the new coordinates.
We componentwise split the above conditions for~$\mathcal T_*$ and thus derive a system of differential equations for the components of~$\mathcal T$,
\begin{gather*}
T_t=\varepsilon,\quad X_t=Y_t=U_t=V_t=0;\\
X_x=a,\quad Y_x=c,\quad T_x=U_x=V_x=0;\\
X_y=b,\quad Y_y=d,\quad T_y=U_y=V_y=0;\\
tX_x+X_u=\varepsilon aT,\quad tY_x+Y_u=\varepsilon cT,\quad U_u=\varepsilon a,\quad V_u=\varepsilon c,\quad T_u=0;\\
tX_y+X_v=\varepsilon bT,\quad tY_y+Y_v=\varepsilon dT,\quad U_v=\varepsilon b,\quad V_v=\varepsilon d,\quad T_v=0;\\
tT_t=T,\quad xX_x+yX_y-uX_u-vX_v=X,\quad xY_x+yY_y-uY_u-vY_v=Y,\\
\qquad uU_u+vU_v=U,\quad vV_v+vV_v=V;\\
t^2T_t=\varepsilon T^2,\quad
txX_x+tyX_y+(x-tu)X_u+(y-tv)X_v=\varepsilon TX,\\
\qquad txY_x+tyY_y+(x-tu)Y_u+(y-tv)Y_v=\varepsilon TY,\\
\qquad (x-tu)U_u+(y-tv)U_v=\varepsilon(X-TU),\quad (x-tu)V_u+(y-tv)V_v=\varepsilon(Y-TV).
\end{gather*}
This system implies that $T=\varepsilon t$ and hence $X_u=Y_v=0$.
Furthermore, $X=ax+by$, $Y=cx+dy$, $U=\varepsilon au+\varepsilon bv$ and $V=\varepsilon cu+\varepsilon dv$.

Note that the system~$\bar{\mathcal S}$ admits the obvious symmetry~$I_0$ permuting~$(x,u)$ and~$(y,v)$,
$I_0\colon(t,x,y,u,v)\mapsto(t,y,x,v,u)$.
Using the chain rule, we express transformed derivatives in terms of the initial coordinates
and substitute the obtained expressions into the copy of the system~$\bar{\mathcal S}$ in the new coordinates.
The expanded system should vanish for each solution of~$\bar{\mathcal S}$.
Applying this procedure to the equation~\eqref{eq:BurgersSystemDiffConstraintForSingle(1+2)DPDE},
we obtain that either $(a,d)=(0,0)$ or $(b,c)=(0,0)$,
and the first case is reduced by~$I_0$ to the second one.
Further we assume that $(b,c)=(0,0)$ and thus $ad\ne0$.
These conditions imply the equations $\varepsilon=1$, $a^2=1$ and $d^2=1$, i.e., $a,d=\pm1$.
This leads to three symmetry transformations alternating signs,
$I_1\colon(t,x,y,u,v)\mapsto(t,-x,y,-u,v)$,
$I_2\colon(t,x,y,u,v)\mapsto(t,x,-y,u,-v)$ and
$I_3\colon(t,x,y,u,v)\mapsto(t,-x,-y,-u,-v)$.
Since $I_2=I_0I_1I_0$ and $I_3=I_1I_2$, the discrete symmetries of the system~~$\bar{\mathcal S}$ are exhausted,
up to combining with continuous symmetries and with each other, by the two involution $I_0$ and~$I_1$,
which do not commute.
\end{proof}

\begin{corollary}
The factor group of the complete point symmetry group~$\bar G$ of the system~$\bar{\mathcal S}$
with respect to its identity component is isomorphic to the  dihedral group~$\mathrm{Dih}_4$ of order~8.
\end{corollary}

Proposition~\ref{pro:BurgersSystemWithDiffConstraintForSingle(1+2)DPDELieSyms} implies
that all point symmetries of the system~$\bar{\mathcal S}$
are induced by point symmetries of the original system~\eqref{eq:BurgersSystem}.
There are no counterparts for rotations among point symmetries of~$\bar{\mathcal S}$,
except rotations by multiples of $\pi/2$, which are discrete symmetries for~$\bar{\mathcal S}$.
Therefore, the Burgers system~\eqref{eq:BurgersSystem} admits no truly conditional symmetries
related to the differential constraint~\eqref{eq:BurgersSystemDiffConstraintForSingle(1+2)DPDE}.

\subsection{Symmetries of potential equation}

\begin{proposition}\label{pro:BurgersSystemSingle(1+2)DPDELieSyms}
The maximal Lie invariance algebra~$\check{\mathfrak g}$
of the equation~\eqref{eq:BurgersSystemSingle(1+2)DPDE} is spanned by the vector fields
\begin{gather*}
\check P^t=\p_t, \quad
\check D=2t\p_t+x\p_x+y\p_y,\quad
\check\Pi=t^2\p_t+tx\p_x+ty\p_y+xy\p_\psi,\quad
\\
\check P^x=\p_x, \quad
\check P^y=\p_y, \quad
\check G^x=t\p_x+y\p_\psi,\quad
\check G^y=t\p_y+x\p_\psi,\quad
\check P^\psi=\p_\psi.
\end{gather*}
The complete point symmetry group~$\check G$ of the equation~\eqref{eq:BurgersSystemSingle(1+2)DPDE}
is generated by one-parameter groups associated with vector fields from the algebra~$\check{\mathfrak g}$
and two discrete transformations: alternating the signs of $(x,\psi)$,
$\check I_1\colon(t,x,y,\psi)\mapsto(t,-x,y,-\psi)$,
and the permutation of~$x$ and~$y$, $\check I_0\colon(t,x,y,\psi)\mapsto(t,y,x,\psi)$.
\end{proposition}

Therefore, the group~$\check G$ consists of the transformations in the space with coordinates $(t,x,y,\psi)$ 
whose $(t,x,y)$-components are of the form~\eqref{txyComponentsOfPOintSymsOfTheArtifitialEq1} or~\eqref{txyComponentsOfPOintSymsOfTheArtifitialEq2} 
and whose $\psi$-component are respectively of the form 
\begin{gather*}
\tilde\psi
=\varepsilon_1\varepsilon_2\psi
-\frac{\gamma\varepsilon_1\varepsilon_2}{\gamma t+\delta}xy
+\frac{\sigma\varepsilon_1\mu_2}{\gamma t+\delta}x
+\frac{\sigma\varepsilon_2\mu_1}{\gamma t+\delta}y
+\frac{\alpha t+\beta}{\gamma t+\delta}\mu_1\mu_2+\varkappa
\\
\hspace*{-\mathindent}\mbox{or}
\\
\tilde\psi
=\varepsilon_1\varepsilon_2\psi
-\frac{\gamma\varepsilon_1\varepsilon_2}{\gamma t+\delta}xy
+\frac{\sigma\varepsilon_1\mu_2}{\gamma t+\delta}y
+\frac{\sigma\varepsilon_2\mu_1}{\gamma t+\delta}x
+\frac{\alpha t+\beta}{\gamma t+\delta}\mu_1\mu_2+\varkappa.
\end{gather*}
Here the description of all the group parameter is the same as 
for~\eqref{txyComponentsOfPOintSymsOfTheArtifitialEq1} and~\eqref{txyComponentsOfPOintSymsOfTheArtifitialEq2}, 
except the additional parameter~$\varkappa$, which is an arbitrary constant.

\begin{proof}
The computation of the algebra~$\check{\mathfrak g}$ is standard and can be realized
using a package for symbolic calculations of Lie symmetries, e.g., the package \texttt{DESOLV}~\cite{carm2000a,vu2007a}.
The Levi decomposition of the algebra~$\check{\mathfrak g}$ is
$\check{\mathfrak g}=\check{\mathfrak f}\lsemioplus\check{\mathfrak r}$,
where the subalgebra $\check{\mathfrak f}=\langle\check P^t,\check D,\check \Pi\rangle$ is a Levi factor of~$\check{\mathfrak g}$,
and the radical~$\check{\mathfrak r}=\langle\check P^x,\check P^y,\check G^x,\check G^y,\check P^\psi\rangle$ of~$\check{\mathfrak g}$
coincides with its nilradical.
This is why, in the course of computing the complete point symmetry group~$\check G$
of the equation~\eqref{eq:BurgersSystemSingle(1+2)DPDE}
up to factoring out Lie symmetries,
similarly to the proofs of Theorem~\ref{thm:BurgersSystemCompletePointSymGroup}
and Proposition~\ref{pro:BurgersSystemWithDiffConstraintForSingle(1+2)DPDELieSyms}
it suffices to consider automorphisms of~$\check{\mathfrak g}$
that preserve the Levi factor~$\check{\mathfrak f}$
and the matrices of whose restrictions to~$\check{\mathfrak f}$ are of the form
$\mathop{\rm diag}(\varepsilon,1,\varepsilon)$ with $\varepsilon=\pm1$.
The set of such automorphisms is exhausted by those with matrices of the form
\[
A=\mathop{\rm diag}(\varepsilon,1,\varepsilon)
\oplus\begin{pmatrix}a&b\\c&d\end{pmatrix}
\oplus\varepsilon\begin{pmatrix}a&b\\c&d\end{pmatrix}
\oplus(\delta),
\]
where $\varepsilon=\pm1$, $ad-bc\ne0$ and $\delta\ne0$.
If the pushforward~$\mathcal T_*$ of vector fields by the point transformation
$\mathcal T\colon (\tilde t,\tilde x,\tilde y,\tilde\psi)=(T,X,Y,\Psi)(t,x,y,\psi)$
is the automorphism of~$\check{\mathfrak g}$ with the matrix~$A$, then
\begin{gather*}
\mathcal T_*\check P^\psi=\delta\tilde P^\psi,\quad
\mathcal T_*\check P^t=\varepsilon\tilde P^t,\quad
\mathcal T_*\check P^x=a\tilde P^x+c\tilde P^y,\quad
\mathcal T_*\check P^y=b\tilde P^x+d\tilde P^y,\\
\mathcal T_*\check G^x=\varepsilon a\tilde G^x+\varepsilon c\tilde G^y,\quad
\mathcal T_*\check G^y=\varepsilon b\tilde G^x+\varepsilon d\tilde G^y,\quad
\mathcal T_*\check D  =\tilde D,\quad
\mathcal T_*\check \Pi=\varepsilon\tilde\Pi,
\end{gather*}
where tildes instead of checks over vector fields mean that these vector fields are given in the new coordinates.
Componentwise splitting of the above conditions for~$\mathcal T_*$ gives a system of differential equations for the components of~$\mathcal T$,
\begin{gather*}
\Psi_\psi=\delta,\quad T_\psi=X_\psi=Y_\psi=0;\quad
T_t=\varepsilon,\quad X_t=Y_t=\Psi_t=0;\\
X_x=a,\quad Y_x=c,\quad T_x=\Psi_x=0;\quad
X_y=b,\quad Y_y=d,\quad T_y=\Psi_y=0;\\
tX_x=\varepsilon aT,\quad tY_x=\varepsilon cT,\quad y\Psi_\psi=\varepsilon aY+\varepsilon cX;\\
tX_y=\varepsilon bT,\quad tY_y=\varepsilon dT,\quad x\Psi_\psi=\varepsilon bY+\varepsilon dX;\\
tT_t=T,\quad xX_x+yX_y=X,\quad xY_x+yY_y=Y;\\
t^2T_t=\varepsilon T^2,\quad txX_x+tyX_y=\varepsilon TX,\quad txY_x+tyY_y=\varepsilon TY,\quad xy\Psi_\psi=\varepsilon XY.
\end{gather*}
The general solution of this system is $T=\varepsilon t$, $X=ax+by$, $Y=cx+dy$ and, up to shifts of~$\Psi$, $\Psi=\delta\psi$,
where the constant parameters~$a$, $b$, $c$, $d$ and~$\delta$ satisfy the equations
$ac=bd=0$ and $\varepsilon(ad+bc)=\delta$.
In view of the inequality $ad-bc\ne0$,
the equations $ac=bd=0$ imply the condition that either $(a,d)=(0,0)$ or $(b,c)=(0,0)$.
The first case is reduced by~$\check I_0$ to the second one,
where~$\check I_0$ is the obvious symmetry of~\eqref{eq:BurgersSystemSingle(1+2)DPDE}
that permutes~$x$ and~$y$, $\check I_0\colon(t,x,y,\psi)\mapsto(t,y,x,\psi)$.
Then the invariance of the equation~\eqref{eq:BurgersSystemSingle(1+2)DPDE}
with respect to the transformation~$\mathcal T$ gives more constraints for the transformation parameters,
which imply, under the assumption $(b,c)=(0,0)$, that $\varepsilon=1$, $a^2=1$, $d^2=1$ and thus $\delta=ad$.
The further proof is similar to the proof of Proposition~\ref{pro:BurgersSystemWithDiffConstraintForSingle(1+2)DPDELieSyms}.
\end{proof}

After the prolongation of vector fields from the algebra~$\check{\mathfrak g}$ and the discrete symmetries~$\check I_0$ and~$\check I_1$
to the derivatives~$\psi_x$ and~$\psi_y$,
it is obvious that all point symmetries of the system~$\bar{\mathcal S}$
are induced by point symmetries of the equation~\eqref{eq:BurgersSystemSingle(1+2)DPDE}.
Moreover, each point symmetry of the equation~\eqref{eq:BurgersSystemSingle(1+2)DPDE}
induces a point symmetry of the system~$\bar{\mathcal S}$,
and the element~$\check P^\psi$ of~$\check{\mathfrak g}$ induces the zero element of~$\bar{\mathfrak g}$.
In other words, $\bar{\mathfrak g}=\mathfrak g'\simeq\check{\mathfrak g}/\langle\check P^\psi\rangle$
and $\bar G=\check G/\check G_0$, where $\check G_0$ is the subgroup of~$\check G$ constituted by the shifts of~$\psi$.
This means that the system~$\bar{\mathcal S}$ admits no potential symmetries with the potential~$\psi$.
Note that the equation~\eqref{eq:BurgersSystemSingle(1+2)DPDE} possesses no local conservation laws,
and hence the system~$\bar{\mathcal S}$ admits only the independent local conservation law
related to the conserved form of the equation~\eqref{eq:BurgersSystemDiffConstraintForSingle(1+2)DPDE},
which is used for introducing the potential~$\psi$ and thus disappears in the course of this,
cf.\ Section~\ref{sec:BurgersSystemCLs}.

Nevertheless, the equation~\eqref{eq:BurgersSystemSingle(1+2)DPDE} is useful
for finding exact solutions of the Burgers system~\eqref{eq:BurgersSystem}
since a single equation~$\mathcal L$ is in general simpler to solve than a system of equations formally alike~$\mathcal L$ in form.
The rest of the present Section~\ref{sec:OneMoreReductionToSingle(1+2)DPDE} illustrates this claim.

\subsection{Essential Lie reduction of potential equation}

We carry out the Lie reduction of the equation~\eqref{eq:BurgersSystemSingle(1+2)DPDE}
with respect to the subalgebra $\check{\mathfrak g}_{2.1}=\langle\check P^t,\check D+2\varsigma\check P^\psi\rangle$
of the maximal Lie invariance algebra~$\check{\mathfrak g}$ of this equation, 
where up to $\check G$-equivalence the constant~$\varsigma$ can be assume nonnegative. 
This is the only codimension-two Lie reduction of the equation~\eqref{eq:BurgersSystemSingle(1+2)DPDE}
that is useful for finding exact solutions of the Burgers system~\eqref{eq:BurgersSystem}
since other such reductions lead to solutions of the Burgers system~\eqref{eq:BurgersSystem}
that are $G$-equivalent to $\mathfrak g^{1.8}$-invariant solutions,
and the whole set of later solutions is described in Section~\ref{sec:SymmetryAnalysisOfReducedSystemsOfPDEs}
in terms of two arbitrary solutions of the (1+1)-dimensional linear heat equation.
We choose the ansatz $\psi=\varphi(\omega)+2\varsigma\ln|x|$ with $\omega=x/y$,
cf.\ Section~\ref{sec:SolutionOfReducedSystemsOfODEs}.
The substitution of this ansatz into the equation~\eqref{eq:BurgersSystemSingle(1+2)DPDE}
gives the second-order reduced ODE
\[
(1+\omega^2)\varphi_{\omega\omega}+(2\omega+2\varsigma)\varphi_\omega+\omega(\varphi_\omega)^2=2\varsigma\omega^{-2},
\]
which is a Riccati equation with respect to~$\varphi_\omega$ and is thus further reduced by the substitution
\[
\varphi_\omega=\frac{1+\omega^2}\omega\frac{\theta_\omega}\theta\quad\mbox{with}\quad \theta=\theta(\omega)
\]
to the linear second-order ODE
\[
\omega(1+\omega^2)^2\theta_{\omega\omega}+(1+\omega^2)(3\omega^2+2\varsigma\omega-1)\theta_\omega-2\varsigma\theta=0.
\]
This equation can be integrated in elementary functions for all values of the parameter~$\varsigma$.
Splitting into different cases depending on~$\varsigma$,
we present its general solution jointly with the corresponding value of~$\varphi_\omega$
(below $\zeta:=\arctan\omega$):
\begin{gather*}
\theta=\frac{e^{-\varsigma\zeta}}{\sqrt{\omega^2+1}}\big(C_1(\omega-\varsigma)+C_2(\zeta(\omega-\varsigma)+1)\big),\quad
\varphi_\omega=\frac2\omega\frac{C_1+C_2(\zeta-\varsigma)}{C_1(\omega-\varsigma)+C_2(\zeta(\omega-\varsigma)+1)}\\
\mbox{if}\quad\varsigma=\pm1;
\\[1ex]
\theta=C_1e^{-\nu_1\zeta}\frac{\omega-\nu_1}{\sqrt{\omega^2+1}}+C_2e^{-\nu_2\zeta}\frac{\omega-\nu_2}{\sqrt{\omega^2+1}},\quad
\varphi_\omega=2\frac\varsigma\omega\frac{C_1\nu_1e^{-\nu_1\zeta}+C_2\nu_2e^{-\nu_2\zeta}}{C_1e^{-\nu_1\zeta}(\omega-\nu_1)+C_2e^{-\nu_2\zeta}(\omega-\nu_2)}\\
\mbox{with}\quad\nu_1=\varsigma+\sqrt{\varsigma^2-1},\quad \nu_2=\varsigma-\sqrt{\varsigma^2-1}
\quad\mbox{if}\quad|\varsigma|>1;
\\[1ex]
\theta=C_1\frac{e^{-\varsigma\zeta}}{\sqrt{\omega^2+1}}\big((\omega-\varsigma)\cos(\mu\zeta)-\mu\sin(\mu\zeta)\big)
      +C_2\frac{e^{-\varsigma\zeta}}{\sqrt{\omega^2+1}}\big((\omega-\varsigma)\sin(\mu\zeta)+\mu\cos(\mu\zeta)\big),\\
\varphi_\omega=2\frac\varsigma\omega\frac
{C_1\big(\varsigma\cos(\mu\zeta)+\mu\sin(\mu\zeta)\big)+C_2\big(\varsigma\sin(\mu\zeta)-\mu\cos(\mu\zeta)\big)}
{C_1\big((\omega-\varsigma)\cos(\mu\zeta)-\mu\sin(\mu\zeta)\big)
+C_2\big((\omega-\varsigma)\sin(\mu\zeta)+\mu\cos(\mu\zeta)\big)}\\
\mbox{with}\quad\mu=\sqrt{1-\varsigma^2}
\quad\mbox{if}\quad|\varsigma|<1.
\end{gather*}
We do not need to integrate the expressions for~$\varphi_\omega$ since
the associated solutions of the Burgers system~\eqref{eq:BurgersSystem} take the form
\[
u=-\frac x{y^2}\varphi_\omega,\quad
v=\frac1y\varphi_\omega+\frac\varsigma x\quad\mbox{with}\quad\omega=\frac xy.
\]

Note that the ansatz constructed with respect to the subalgebra $\check{\mathfrak g}_{2.1}$
induces the representation $\varphi^1=-\omega\varphi_\omega$ and $\varphi^2=\varphi_\omega+\varsigma\omega^{-1}$
for the dependent invariant functions of the ansatz~\eqref{eq:BurgersSystemAnzatz2.1kappa0}, 
which equivalent to the reduction 
of the differential constraint~\eqref{eq:BurgersSystemDiffConstraintForSingle(1+2)DPDE} 
by the ansatz~\eqref{eq:BurgersSystemAnzatz2.1kappa0} 
to the equation $(\varphi^1+\omega\varphi^2)_\omega=0$.
It is quite unobvious that this representation helps to construct
particular solutions of the reduced system~\eqref{eq:BurgersSystemNewReducedSystem2.1kappa0}.
The equivalent ansatz~2.1 with $\kappa=0$ reduces
the differential constraint~\eqref{eq:BurgersSystemDiffConstraintForSingle(1+2)DPDE}
to the equation
$\cos(2\omega)(\varphi^2_\omega+2\varphi^1)+\sin(2\omega)(\varphi^1_\omega-2\varphi^2)=0$,
which is even less obvious than the above representation.

\subsection{Solutions of potential equation that are affine in a space variable}

If a solution~$\psi$ of the equation~\eqref{eq:BurgersSystemSingle(1+2)DPDE} is affine in a space variable,
up to the permutation of~$x$ and $y$, which is a point symmetry of~\eqref{eq:BurgersSystemSingle(1+2)DPDE},
we can assume that $\psi$ is affine in~$y$, i.e., $\psi=\psi^1(t,x)y+\psi^0(t,x)$.
The set of such solutions is obviously singled out by the differential constraint $\psi_{yy}=0$.
Substituting the above representation for~$\psi$ into the equation~\eqref{eq:BurgersSystemSingle(1+2)DPDE}
and splitting with respect to~$y$,
we obtain a system of (1+1)-dimensional PDEs for~$\psi^1$ and~$\psi^0$,
which coincides, up to notation of independent and dependent variables, with reduced system~1.8,
\begin{gather*}
\psi^1_t+\psi^1\psi^1_x-\psi^1_{xx}=0,\\
\psi^0_t+\psi^1\psi^0_x-\psi^0_{xx}=0.
\end{gather*}
Therefore, the Hopf--Cole-type transformation $\psi^1=-2\theta^1_x/\theta^1$, $\psi^0=\theta^0/\theta^1$,
cf.~\eqref{eq:Hopf--Cole-likeTransForReducedSystem1.8},
reduces this system to the system of two copies of the (1+1)-dimensional linear heat equation,
$\theta^1_t=\theta^1_{xx}$, $\theta^0_t=\theta^0_{xx}$.
In view of the ansatz~\eqref{eq:BurgersSystemSolutionsViaSingle(1+2)DPDE}
reducing the Burgers system~\eqref{eq:BurgersSystem} to the equation~\eqref{eq:BurgersSystemSingle(1+2)DPDE},
this leads to one more subset of solutions of the Burgers system~\eqref{eq:BurgersSystem},
which are expressed in terms of a pair of solutions of the (1+1)-dimensional linear heat equation:
\begin{gather}\label{eq:BurgersSystem2ndSolutionSetIn2SolutionsOfLHEq}
u=-2\frac{\theta^1_x}{\theta^1}, \quad
v=\left(-2\frac{\theta^1_x}{\theta^1}y+\frac{\theta^0}{\theta^1}\right)_x
\quad\mbox{with}\quad \theta^i=\theta^i(t,x)\colon\ \theta^i_t=\theta^i_{xx},\quad i=0,1.
\end{gather}

\subsection{Solutions related to a complex Hamilton--Jacobi equation}\label{sec:SolutionsRelatedToComplexHamiltonJacobiEq}

An interesting solution subset of the equation~\eqref{eq:BurgersSystemSingle(1+2)DPDE}
is constituted by harmonic solutions, which additionally satisfy the Laplace equation $\psi_{xx}+\psi_{yy}=0$.
The corresponding overdetermined system
\begin{gather}\label{eq:BurgersSystemSingle(1+2)DPDESystemForharmonicSolutions}
\psi_t+\psi_x\psi_y=0,\quad \psi_{xx}+\psi_{yy}=0
\end{gather}
completed the trivial differential consequences of its first equation,
which are obtained by single differentiations with respect to~$t$, $x$ and~$y$,
is formally compatible,
and so no tools of the formal compatibility theory are relevant here.
To describe the general solution of this system, we use a different approach.
Since the function~$\psi$ is harmonic, it can be represented as, e.g.,
the real part of a complex-valued function $f=f(t,z)$
that is smooth with respect to~$t$ and holomorphic (i.e., complex analytic) with respect to~$z=x+iy$,
\[
\psi=\mathop{\rm Re}f(t,z)=\frac12(f(t,z)+f^*(t,z^*)).
\]
Here $i$ is the imaginary unit, asterisks denote the complex conjugation, and hence $z^*=x-iy$.
Under the above representation,
the Laplace equation $\psi_{xx}+\psi_{yy}=0$ can be neglected,
and the equation $\psi_t+\psi_x\psi_y=0$ reduces to
\[
f_t+f^*_t+\frac i2(f_z+f^*_{z^*})(f_z-f^*_{z^*})=0.
\]
We can separate the variables~$z$ and~$z^*$ in the last equation, obtaining
\[
f_t+\frac i2(f_z)^2=-f^*_t+\frac i2(f^*_{z^*})^2=i\lambda.
\]
Here $\lambda$ is a real-valued function depending only on~$t$.
It suffices to consider only the equation
\begin{gather}\label{eq:BurgersSystemSingle(1+2)DPDEHamiltonJacobiEq}
f_t+\frac i2(f_z)^2=i\lambda
\end{gather}
since its counterpart for~$f^*$ is the complex conjugate of it.
We can interpret the equation~\eqref{eq:BurgersSystemSingle(1+2)DPDEHamiltonJacobiEq} 
as a complex Hamilton--Jacobi equation.
Its complete integral is $f=az-\frac i2a^2t+i\Lambda(t)+b$,
where $a$ and~$b$ are complex constants,
and $\Lambda$ is a fixed antiderivative of~$\lambda$.
The equation~\eqref{eq:BurgersSystemSingle(1+2)DPDEHamiltonJacobiEq}
admits no singular integrals,
and the parametric representation of its general integral is
\begin{gather}\label{eq:BurgersSystemSingle(1+2)DPDEHamiltonJacobiEqGenIntegral}
f=az-\frac i2a^2t+i\Lambda(t)+F(a),\quad z-iat+F_a(a)=0,
\end{gather}
where $F=F(a)$ is an arbitrary holomorphic function of~$a$.
Solving the second equation of~\eqref{eq:BurgersSystemSingle(1+2)DPDEHamiltonJacobiEqGenIntegral} with a fixed~$F$ for~$a$
and substituting the obtained expression into the first equation,
we obtain a solution of the equation~\eqref{eq:BurgersSystemSingle(1+2)DPDEHamiltonJacobiEq}.
The corresponding solution of the Burgers system~\eqref{eq:BurgersSystem} is given~by
\[
u=(\mathop{\rm Re}f)_y=\frac i2(f_z-f^*_{z^*})=-\mathop{\rm Im}f_z,\quad
v=(\mathop{\rm Re}f)_x=\frac 12(f_z+f^*_{z^*})= \mathop{\rm Re}f_z.
\]
Note that varying the parameter function~$\lambda$ (resp.\ $\Lambda$)
has no influence on the expressions for~$\psi$, $u$ and~$v$.
Therefore, this parameter function can be assume to vanish.

If $F$ is a polynomial with $\deg F\leqslant3$, then we can derive explicit expressions
for the associated functions~$f$, $u$ and~$v$.
Thus, in the case where $\deg F\leqslant2$ we construct solutions of the Burgers system~\eqref{eq:BurgersSystem}
that are affine in the space variables~$(x,y)$ and take, up to shifts of the independent variables, the form
\begin{gather}\label{eq:BurgersSystemParticularSolutionAffineInXY}
u=\frac{tx-\mu y}{t^2+\mu^2},\quad
v=\frac{\mu x+ty}{t^2+\mu^2},
\end{gather}
where $\mu$ is an arbitrary (real) constant, and $\mu\in\{0,1\}\bmod G$.
The solution~\eqref{eq:BurgersSystemParticularSolutionAffineInXY} with $\mu=1$
is the simplest $\mathfrak g^{2.5}_0$- and $\mathfrak g^{2.6}_1$-invariant solution,
which is obtained by setting $\varphi^1=\varphi^2=0$ in ansatzes~2.5 and~2.6.
In fact, it is invariant with respect to the subalgebra $\langle P^t+\Pi+J,G^y+P^x,G^x-P^y\rangle$ of~$\mathfrak g$.
The solution~\eqref{eq:BurgersSystemParticularSolutionAffineInXY} with $\mu=0$
is invariant with respect to the subalgebra $\langle D,\Pi,J,G^x,G^y\rangle$.
See Section~\ref{sec:SolutionsAffineInSpaceVars} below 
for the description of the entire set of solutions of the Burgers system~\eqref{eq:BurgersSystem}
that are affine in the space variables~$(x,y)$.

Suppose now that $\deg F=3$.
We represent~$F$ in the form
$F=\frac13\alpha a^3+\frac12\beta a^2+\gamma a+\delta$.
In fact, up symmetry transformations of the original system~\eqref{eq:BurgersSystem}
we can set $\alpha=1$, $\beta\in\mathbb R$ and $\gamma=0$.
This gives the following solution of the Burgers system~\eqref{eq:BurgersSystem}:
\begin{gather*}
u=-\frac t   2\pm\frac12\sqrt{\frac{\sqrt{\zeta^2+\theta^2}-\zeta}2}\sgn\theta,\quad
v=-\frac\beta2\pm\frac12\sqrt{\frac{\sqrt{\zeta^2+\theta^2}+\zeta}2},\quad
\end{gather*}
where $\zeta=\beta^2-t^2-4x$, $\theta=2\beta t+4y$, $\beta\in\mathbb R$. Note that $\beta=0\bmod G$.
%The sign of theta is changed !!!

\section{Solutions affine in space variables}\label{sec:SolutionsAffineInSpaceVars}

The solutions of the form~\eqref{eq:BurgersSystemParticularSolutionAffineInXY},
even extended by transformations from the group~$G$,
do not exhaust the set ${\rm AS}_{\eqref{eq:BurgersSystem}}$ of solutions of the Burgers system~\eqref{eq:BurgersSystem}
that are affine in the space variables~$(x,y)$.
This set is $G$-invariant and is singled out from the whole solution set of~\eqref{eq:BurgersSystem}
by the differential constraint
$
u_{xx}=u_{xy}=u_{yy}=v_{xx}=v_{xy}=v_{yy}=0.
$
The general form of solutions from~${\rm AS}_{\eqref{eq:BurgersSystem}}$ is
\begin{gather}\label{eq:BurgersSystemAffineAnsatz}
{\bf u}=A(t){\bf x}+{\bf b}(t),
\end{gather}
where ${\bf u}=(u,v)^{\mathsf T}$, ${\bf x}=(x,y)^{\mathsf T}$ and
$A=A(t)$ (resp.\ ${\bf b}={\bf b}(t)$) is a smooth $2\times2$
(resp.\ $2\times1$) matrix function of~$t$.
We substitute the ansatz~\eqref{eq:BurgersSystemAffineAnsatz} into the system~\eqref{eq:BurgersSystem}
and split resulting equations with respect to~$x$ and~$y$,
which gives equations for~$A$ and~${\bf b}$,
\begin{gather}\label{eq:BurgersSystemAffineAnsatzSystemForParameters}
A_t+A^2=0,\quad {\bf b}_t+A{\bf b}=0.
\end{gather}
We separately consider different cases of integrating the system~\eqref{eq:BurgersSystemAffineAnsatzSystemForParameters}
depending on the degeneration degree of the matrix~$A$.
Each of these cases is $G$-invariant.

\medskip\par\noindent
1.\ $\det A\ne0$. We rewrite the first equation of~\eqref{eq:BurgersSystemAffineAnsatzSystemForParameters}
in the form $-A^{-1}A_tA^{-1}=E$, i.e., $(A^{-1})_t=E$, which integrates to $A^{-1}=tE+C$ and thus $A=(tE+C)^{-1}$.
Here $E$ is the $2\times2$ identity matrix and $C$ is an arbitrary constant $2\times2$ matrix.
In view of the first equation of~\eqref{eq:BurgersSystemAffineAnsatzSystemForParameters} and the nondegeneracy of~$A$,
the general solution of the second equation takes the form ${\bf b}=A{\bf b}^0$,
where ${\bf b}^0$ is an arbitrary constant $2\times1$ matrix.
The substitution of the obtained expressions into the ansatz~\eqref{eq:BurgersSystemAffineAnsatz}
leads to the following family of solutions from~${\rm AS}_{\eqref{eq:BurgersSystem}}$:
\[
{\bf u}=(tE+C)^{-1}({\bf x}+{\bf b}^0),
\]
where, in particular, ${\bf b}^0=0$ and $\mathop{\rm tr}C=0$ modulo $G$-equivalence.
Solutions of the form~\eqref{eq:BurgersSystemParticularSolutionAffineInXY} belong to this family.

\medskip\par\noindent
2.\ $\det A=0$, $\mathop{\rm tr}A\ne0$.
Up to permutation of~$(x,u)$ and~$(y,v)$, we can represent the matrix~$A$ as
\[
A=\begin{pmatrix}\alpha&\beta\\ \gamma&\beta\gamma/\alpha\end{pmatrix},
\]
\looseness=-1
where the entries $\alpha$, $\beta$ and~$\gamma$ are smooth functions of~$t$ with $\alpha\ne0$.
Then the matrix equation $A_t=-A^2$ expands to the system
$\alpha_t=-\lambda\alpha$, $\beta_t=-\lambda\beta$, $\gamma_t=-\lambda\gamma$
with $\lambda=\alpha+\beta\gamma/\alpha$.
The equation for the $(2,2)$-entry is identically satisfied in view of this system.
It also implies that $(\beta/\alpha)_t=0$ and $(\gamma/\alpha)_t=0$, i.e.,
$\beta=\alpha c_1$ and $\gamma=\alpha c_2$ for some constants~$c_1$ and~$c_2$.
Hence the equation for~$\alpha$ reduces to $\alpha_t=-(1+c_1c_2)\alpha^2$,
where $1+c_1c_2\ne0$ since $\mathop{\rm tr}A\ne0$,
and nonvanishing solutions of the last equation are given by
$\alpha=((1+c_1c_2)t+\kappa)^{-1}$, where $\kappa$ is an arbitrary constant.
Solving the second equation of~\eqref{eq:BurgersSystemAffineAnsatzSystemForParameters},
we get ${\bf b}=\mu\alpha(1,c_2)^{\mathsf T}+\nu(-c_1,1)^{\mathsf T}$,
where $\mu$ and~$\nu$ are arbitrary constants.
We substitute the expressions obtained for~$A$ and~$b$ into the ansatz~\eqref{eq:BurgersSystemAffineAnsatz}
and symmetrize the resulting expression for~${\bf u}$,
which gives
\[
{\bf u}=\frac1{t\mathop{\rm tr}C+\kappa}C(x+{\bf m})+{\bf n},
\]
where $\kappa$ is an arbitrary constant,
$C$ is an arbitrary constant $2\times2$ matrix with $\det C=0$ and $\mathop{\rm tr}C\ne0$,
${\bf m}$ and~${\bf n}$ are arbitrary constant $2\times1$ matrices with $C{\bf n}=0$.
Up to $G$-equivalence, we can assume that
$\kappa=0$, ${\bf m}={\bf n}={\bf 0}$, $\mathop{\rm tr}C=1$
and the second row of~$C$ is zero, i.e.,
$u=t^{-1}(x+c_1y)$, $v=0$, where $c_1$ is an arbitrary constant.

\medskip\par\noindent
3.\ $\det A=\mathop{\rm tr}A=0$, $A\ne0$.
Hence $A^2=0$, and thus $A_t=0$, i.e., $A$ is a constant nilpotent matrix.
Then the second equation of~\eqref{eq:BurgersSystemAffineAnsatzSystemForParameters}
integrates to ${\bf b}=\mu(A{\bf b}^1t+{\bf b}^1)+\nu A{\bf b}^1$,
where $\mu$ and~$\nu$ are arbitrary constants
and ${\bf b}^1$ is a fixed constant $2\times1$ matrix with $A{\bf b}^1\ne{\bf 0}$.
The corresponding solution of the system~\eqref{eq:BurgersSystem} is
\[
{\bf u}=A({\bf x}+\mu{\bf b}^1t+\nu{\bf b}^1)-\mu{\bf b}^1,
\]
and up to  $G$-equivalence $A=\begin{pmatrix}0&1\\0&0\end{pmatrix}$ and $\mu=\nu=0$,
i.e., $u=y$ and $v=0$.

\medskip\par\noindent
4.\ $A=0$. In this case we have only trivial constant solutions.
Moreover, any constant solution is $G$-equivalent to the zero solution
with respect to Galilean boosts.

\medskip\par
It can be checked that each solution affine in the space variables~$(x,y)$
is invariant with respect to a three-dimensional subalgebra
of the maximal Lie invariance algebra~$\mathfrak g$ of the Burgers system~\eqref{eq:BurgersSystem}.

\section{Common solutions of `viscid' and  `inviscid' Burgers systems}
\label{sec:CommonSolutionsOfViscidAndInviscidBurgersSystems}

The classical Burgers equation $u_t+uu_x=\nu u_{xx}$ and its `inviscid' counterpart $u_t+uu_x=0$
admit only trivial common solutions that are at most affine in~$x$.
The set of common solutions is exhausted by two solution families~\cite{poch2013a},
the one-parameter family of constant solutions
and the two-parameter family of solutions of the form $u=(x+C_1)/(t+C_2)$,
where $C_1$ and $C_2$ are arbitrary constants.
For the `viscid' Burgers system~\eqref{eq:BurgersSystem} and the `inviscid' Burgers system
\begin{gather}\label{eq:InvicidBurgersSystem}
\begin{split}
&u_t+uu_x+vu_y=0,\\
&v_t+uv_x+vv_y=0,
\end{split}
\end{gather}
the structure of the set of their common solutions is not so obvious.
Additionally to~\eqref{eq:InvicidBurgersSystem}, both the components~$u$ an~$v$ of these solutions
satisfy the two-dimensional Laplace equation, $u_{xx}+u_{yy}=0$ and $v_{xx}+v_{yy}=0$,
and hence admit the representation
$u=\mathop{\rm Re}f=\frac12(f+f^*)$, $v=\mathop{\rm Im}g=\frac12(g+g^*)$,
where $f=f(t,z)$ and $g=g(t,z)$ are complex-valued functions 
that are smooth with respect to~$t$ and holomorphic (i.e., complex analytic) with respect to~$z=x+iy$.
Here again $i$ is the imaginary unit, and asterisks denote the complex conjugation.
We substitute the above representation into the system~\eqref{eq:InvicidBurgersSystem}, getting
\begin{gather}\label{eq:InvicidBurgersSystemComplexification}
\begin{split}
&2(f_t+f^*_t)+(f+f^*)(f_z+f^*_{z^*})+i(g+g^*)(f_z-f^*_{z^*})=0,\\
&2(g_t+g^*_t)+(f+f^*)(g_z+g^*_{z^*})+i(g+g^*)(g_z-g^*_{z^*})=0.
\end{split}
\end{gather}
Consider the differential consequences derived by
differentiating both the equations~\eqref{eq:InvicidBurgersSystemComplexification}
with respect to~$z$ and~$z^*$,
\begin{gather*}
(f^*_{z^*}+ig^*_{z^*})f_{zz}+(f_z-ig_z)f^*_{z^*z^*}=0,\quad
(f^*_{z^*}+ig^*_{z^*})g_{zz}+(f_z-ig_z)g^*_{z^*z^*}=0.
\end{gather*}

The condition $f_z-ig_z=0$ and its conjugate $f^*_{z^*}+ig^*_{z^*}=0$ mean
that the functions~$u$ and~$v$ satisfy the constraint $u_x=v_y$.
Such solutions have exhaustively been studied in Section~\ref{sec:SolutionsRelatedToComplexHamiltonJacobiEq}.

Further suppose that $f_z-ig_z\ne0$.
Then we can separate the variables in the differential consequences,
\begin{gather*}
\frac{f_{zz}}{f_z-ig_z}=-\frac{f^*_{z^*z^*}}{f^*_{z^*}+ig^*_{z^*}}=i\lambda^1,\quad
\frac{g_{zz}}{f_z-ig_z}=-\frac{g^*_{z^*z^*}}{f^*_{z^*}+ig^*_{z^*}}=i\lambda^2,
\end{gather*}
where $\lambda^1$ and~$\lambda^2$ are real-valued smooth functions of~$t$.
Therefore,
\[
f_{zz}=i\lambda^1(f_z-ig_z),\quad
g_{zz}=i\lambda^2(f_z-ig_z),
\]
and these equations are recombined to
$(f-ig)_{zz}=i\lambda^*(f-ig)_z$,
$(f+ig)_{zz}=i\lambda(f-ig)_z$
with $\lambda:=\lambda^1+i\lambda^2$.
If $\lambda\ne0$, solutions of these equations do not satisfy
the system~\eqref{eq:InvicidBurgersSystemComplexification}.
Hence $\lambda=0$, i.e., $f_{zz}=g_{zz}=0$.
In other words, the tuple $(u,v)$ is affine in $(x,y)$ with coefficients depending on~$t$.
The solutions of the Burgers system~\eqref{eq:BurgersSystem} with this property
have been constructed in Section~\ref{sec:SolutionsAffineInSpaceVars},
and all of them also satisfy  the `inviscid' Burgers system~\eqref{eq:InvicidBurgersSystem}.

Summing up, we prove the following assertion.

\begin{proposition}
The set of common solutions of the `viscid' Burgers system~\eqref{eq:BurgersSystem}
and the `inviscid' Burgers system~\eqref{eq:InvicidBurgersSystem}
is the union of two solution subsets.
One of them consists of solutions satisfying,
additionally to the Laplace equations $u_{xx}+u_{yy}=0$ and $v_{xx}+v_{yy}=0$,
the differential constraint $u_x=v_y$.
The other subset is constituted by solutions affine in the space variables~$(x,y)$.
\end{proposition}

Recall that these solution subsets are not disjoint.
Their intersection includes, in particular, the zero solution
and any solution $G$-equivalent to a solution of the form~\eqref{eq:BurgersSystemParticularSolutionAffineInXY}.

\section{More solutions with method of differential constraints}\label{sec:BurgersSystemMoreSolutionsWithMethodOfDiffConstraints}

If a linear combination of~$u$ and~$v$ is a constant, then up to $G$-equivalence we can
assume that $v=0$ and hence the component~$u$ satisfies the (1+2)-dimensional Burgers equation~\eqref{(1+2)DBurgersEq}.
Some exact solutions of this equation
that are not related to the linearizable solution subsets of the system~\eqref{eq:BurgersSystem}
were constructed in~\cite{raja2008a}.

We generalize the corresponding solutions of the system~\eqref{eq:BurgersSystem}.
Suppose that the component~$v$ does not depend on~$x$, $v_x=0$.
Then the second equation of the system~\eqref{eq:BurgersSystem} reduces to the classical Burgers equation,
and thus we obtain the system
\begin{gather}\label{eq:BurgersSystemWithVx=0}
\begin{split}
&u_t+uu_x+vu_y-u_{xx}-u_{yy}=0,\\
&v_t+vv_y-v_{yy}=0,\quad v_x=0.
\end{split}
\end{gather}
Let us additionally set the differential constraint $u_{xx}=0$,
i.e., we assume that $u$ is affine in~$x$, $u=u^1(t,y)x+u^0(t,y)$,
which leads to the partially coupled system
\begin{gather}\label{eq:BurgersSystemWithVx=0AndUxx=0}
\begin{split}
&v_t+vv_y-v_{yy}=0,\\
&u^1_t+u^1u^1+vu^1_y-u^1_{yy}=0,\\
&u^0_t+u^0u^1+vu^0_y-u^0_{yy}=0
\end{split}
\end{gather}
for the functions $u^1$, $u^0$ and $v$ depending on~$(t,y)$.
The reduction of the system~\eqref{eq:BurgersSystem} to the system~\eqref{eq:BurgersSystemWithVx=0AndUxx=0}
by the ansatz $u=u^1(t,y)x+u^0(t,y)$, $v=v(t,y)$ means
that the system~\eqref{eq:BurgersSystem} is conditionally invariant with respect to
the generalized symmetry $u_{xx}\p_u+v_x\p_v$.
The maximal Lie invariance algebra of the system~\eqref{eq:BurgersSystemWithVx=0AndUxx=0}
is spanned by the vector fields
\begin{gather*}
\bar P^t=\p_t, \quad
\bar D=2t\p_t+y\p_y-v\p_v-2u^1\p_{u^1},
\\
\bar\Pi=t^2\p_t+ty\p_y+(y-tv)\p_v-tu^0\p_{u^0}+(1-2tu^1)\p_{u^1},\quad
\bar I=u^0\p_{u^0},
\\
\bar P^x=u^1\p_{u^0}, \quad
\bar P^y=\p_y, \quad
\bar G^x=(1-tu^1)\p_{u^0},\quad
\bar G^y=t\p_y+\p_v.
\end{gather*}
All of the above Lie symmetry vector fields, except~$\bar I$, are respectively induced
by the basis elements of the maximal Lie invariance algebra~$\mathfrak g$
of the system~\eqref{eq:BurgersSystem} with the same notation without bar.
In other words, the Burgers system~\eqref{eq:BurgersSystem} admits nontrivial hidden symmetries
related to the ansatz $u=u^1(t,y)x+u^0(t,y)$, $v=v(t,y)$,
which necessarily involve, as a summand, a vector field proportional to~$\bar I$.
More important is that the entire Lie invariance algebra of the first equation
of the system~\eqref{eq:BurgersSystemWithVx=0AndUxx=0} is prolonged to~$(u^0,u^1)$,
and the prolonged algebra coincides with
$\langle\bar P^t,\bar D,\bar\Pi,\bar P^y,\bar G^y\rangle$.
The only independent discrete symmetry of this equation
that simultaneously alternates the signs of~$y$ and~$v$ is trivially prolonged
to a discrete symmetry of the entire system~\eqref{eq:BurgersSystemWithVx=0AndUxx=0}.
This is why for finding $G$-inequivalent solutions
of the system~\eqref{eq:BurgersSystemWithVx=0AndUxx=0}
by sequentially solving its equations in view of partially coupling of it,
one should start with  solutions of the Burgers equation for~$v$
that are inequivalent with respect to the point symmetry group of this equation.

Under the additional constraint $u^1=v_y$, the second equation of~\eqref{eq:BurgersSystemWithVx=0AndUxx=0}
is a differential consequence of the first one, taking the form $(v_t+vv_y-v_{yy})_y=0$,
and the third equation reduces to the equation $u^0_t+(vu^0-u^0_y)_y=0$,
which is in conserved form.
Using this conserved form, we introduce the potential~$w$,
which is defined by the system $w_y=u^0$, $w_t=-vu^0-u^0_y$
and thus satisfies the equation $w_t+vw_y-w_{yy}=0$.
The system of (1+1)-dimensional PDEs for~$v=v(t,y)$ and~$w=w(t,y)$
coincides, up to notation of independent and dependent variables, with reduced system~1.8.
Therefore, the Hopf--Cole-type transformation $v=-2\theta^1_y/\theta^1$, $w=\theta^0/\theta^1$,
cf.~\eqref{eq:Hopf--Cole-likeTransForReducedSystem1.8},
reduces this system to the system of two copies of the (1+1)-dimensional linear heat equation,
$\theta^1_t=\theta^1_{yy}$, $\theta^0_t=\theta^0_{yy}$.
As a result, we construct the solution family of the Burgers system~\eqref{eq:BurgersSystem}
\begin{gather}\label{eq:BurgersSystem2ndSolutionSetIn2SolutionsOfLHEqModified}
v=\left(-2\frac{\theta^1_y}{\theta^1}x+\frac{\theta^0}{\theta^1}\right)_y, \quad
v=-2\frac{\theta^1_y}{\theta^1}
\quad\mbox{with}\quad \theta^i=\theta^i(t,y)\colon\ \theta^i_t=\theta^i_{yy},\quad i=0,1.
\end{gather}
which differs from the family~\eqref{eq:BurgersSystem2ndSolutionSetIn2SolutionsOfLHEq}
by the permutation of $(x,u)$ and~$(y,v)$.

We can carry out the Hopf--Cole-type transformation
$u=\tilde u/\tilde v$, $v=-2\tilde v_y/\tilde v$
%\[u=\frac{\tilde u}{\tilde v}, \quad v=-2\frac{\tilde v_y}{\tilde v}\]
in~\eqref{eq:BurgersSystemWithVx=0}, to get the system
\begin{gather}
\begin{split}
&\tilde u_t+\frac1{\tilde v}\tilde u\tilde u_x-\tilde u_{xx}-\tilde u_{yy}=0,\\
&\tilde v_t-\tilde v_{yy}=0,\quad \tilde v_x=0.
\end{split}
\end{gather}
Suppose additionally that $\tilde u_{xx}=0$, i.e., $\tilde u$ is affine in~$x$, $\tilde u=w^1(t,y)x+w^0(t,y)$.
As a result, we obtain the system
\begin{subequations}\label{eq:BurgersSystemWithVx=0Uxx=0AndHopf-ColeTrans}
\begin{gather}
\tilde v_t-\tilde v_{yy}=0,
\label{eq:BurgersSystemWithVx=0Uxx=0AndHopf-ColeTrans1}\\
w^1_t-w^1_{yy}+\frac{(w^1)^2}{\tilde v}=0,
\label{eq:BurgersSystemWithVx=0Uxx=0AndHopf-ColeTrans2}\\
w^0_t-w^0_{yy}+\frac{w^1}{\tilde v}w^0=0
\label{eq:BurgersSystemWithVx=0Uxx=0AndHopf-ColeTrans3}
\end{gather}
\end{subequations}
with the two independent variables~$t$ and~$y$.
Lie reduction~1.8 is $G$-equivalent to the case $w^1=0$ giving $\tilde u_x=0$.
An alternative way for deriving the system~\eqref{eq:BurgersSystemWithVx=0Uxx=0AndHopf-ColeTrans}
is to make the Hopf--Cole-type transformation
$u^0=w^0/\tilde v$, $u^1=w^1/\tilde v$, $v=-2\tilde v_y/\tilde v$
in the system~\eqref{eq:BurgersSystemWithVx=0AndUxx=0}.

Similarly to the system~\eqref{eq:BurgersSystemWithVx=0AndUxx=0},
the system~\eqref{eq:BurgersSystemWithVx=0Uxx=0AndHopf-ColeTrans} is partially coupled.
Therefore, finding its exact solutions can be split into three steps:
choosing a solution of the linear heat equation~\eqref{eq:BurgersSystemWithVx=0Uxx=0AndHopf-ColeTrans1},
substituting this value of~$\tilde v$ into the equation~\eqref{eq:BurgersSystemWithVx=0Uxx=0AndHopf-ColeTrans2}
and constructing an exact solution~$w^1$ of the obtained nonlinear (variable-coefficient) heat equation with quadratic nonlinearity,
and, finally, solving the linear heat equation~\eqref{eq:BurgersSystemWithVx=0Uxx=0AndHopf-ColeTrans3}
with the potential $w^1/\tilde v$.
Thus, for any solution $(w^1,\tilde v)$
of the subsystem~\eqref{eq:BurgersSystemWithVx=0Uxx=0AndHopf-ColeTrans1}--\eqref{eq:BurgersSystemWithVx=0Uxx=0AndHopf-ColeTrans2}
we construct the family of solutions
\[
u=\frac{w^1}{\tilde v}x+\frac{w^0}{\tilde v},\quad
v=-2\frac{\tilde v_y}{\tilde v}
\]
of the Burgers system~\eqref{eq:BurgersSystem}
parameterized  by solutions of the linear heat equation~\eqref{eq:BurgersSystemWithVx=0Uxx=0AndHopf-ColeTrans3}
with the potential $w^1/\tilde v$.
We need to select solutions with $w^1\ne-2\tilde v_{yy}+2(\tilde v_y)^2/\tilde v$
since all other solutions belong to the family~\eqref{eq:BurgersSystem2ndSolutionSetIn2SolutionsOfLHEqModified}.
The maximal Lie invariance algebra~$\tilde{\mathfrak g}$ of the system~\eqref{eq:BurgersSystemWithVx=0Uxx=0AndHopf-ColeTrans}
is spanned by the vector fields
\begin{gather*}
\tilde P^t=\p_t, \quad
\tilde D=2t\p_t+y\p_y+2\tilde v\p_{\tilde v},
\\
\tilde\Pi=4t^2\p_t+4ty\p_y-(y^2+2t)\tilde v\p_{\tilde v}-(y^2+6t)w^0\p_{w^0}-(y^2w^1+10tw^1-4v)\p_{w^1},\quad
\\
\tilde P^x=w^1\p_{w^0}, \quad
\tilde P^y=\p_y, \quad
\tilde G^x=(\tilde v-tw^1)\p_{w^0},\quad
\tilde G^y=2t\p_y-y\tilde v\p_{\tilde v}-yw^0\p_{w^0}-yw^1\p_{w^1},
\\
\tilde I^1=\tilde v\p_{\tilde v}+w^1\p_{w^1},\quad
\tilde I^0=w^0\p_{w^0}.
\end{gather*}
Each basis element of~$\tilde{\mathfrak g}$ except $\tilde I^1$ induces, via the first prolongation,
the respective basis element
of the maximal Lie invariance algebra of the system~\eqref{eq:BurgersSystemWithVx=0AndUxx=0};
cf.\ Section~\ref{sec:SymmetryAnalysisOfReducedSystemsOfPDEs}.
The basis element $\tilde I^1$ is mapped to the zero vector field.

The essential invariance algebra (see the definition in \cite{bihl2017a,kuru2016a})
of the single equation \eqref{eq:BurgersSystemWithVx=0Uxx=0AndHopf-ColeTrans1},
which is just the classical linear heat equation, is entirely prolonged to $(w^0,w^1)$,
and the prolonged algebra is equal to
$\langle\tilde P^t,\tilde D,\tilde\Pi,\tilde P^y,\tilde G^y,\tilde I^1\rangle$.
A similar assertion on discrete symmetries is also true.
This is why in the course of the construction of solutions
of the system~\eqref{eq:BurgersSystemWithVx=0Uxx=0AndHopf-ColeTrans}
it suffices to use only solutions of the equation~\eqref{eq:BurgersSystemWithVx=0Uxx=0AndHopf-ColeTrans1}
that are inequivalent with respect to its point symmetry group.

We return to solutions with $v=0$, which corresponds to constant values of~$\tilde v$.
Without loss of generality we can set $\tilde v=1$.
As mentioned in the beginning of this section,
then the component~$u$ satisfies the (1+2)-dimensional Burgers equation~\eqref{(1+2)DBurgersEq}.
We arrange and extend solution families of~\eqref{(1+2)DBurgersEq}
that were constructed in~\cite{raja2008a}.

Looking for stationary solutions of the equation~\eqref{eq:BurgersSystemWithVx=0Uxx=0AndHopf-ColeTrans2}
with $\tilde v=1$, which is the nonlinear heat equation with purely quadratic nonlinearity,
we integrate the system $w^1_t=0$, $w^1_{yy}=(w^1)^2$
whose general solution is
$w^1=\wp\left(y/\sqrt6+C_2;0,C_1\right)$.
Here $C_1$ and~$C_2$ are arbitrary constants
and  $\wp=\wp(z;\textsl{g}_2,\textsl{g}_3)$ is the Weierstrass elliptic function,
satisfying the differential equation
\[(\wp_z)^2 = 4\wp^3-\textsl{g}_2\wp-\textsl{g}_3,\]
where the numbers $\textsl{g}_2$ and $\textsl{g}_3$ called invariants
are respectively equal to $0$ and~$C_1$ for~$w^1$.
For such general values of~$w^1$, the essential invariance algebra of
the equation~\eqref{eq:BurgersSystemWithVx=0Uxx=0AndHopf-ColeTrans3} is $\langle\p_t,w^0\p_{w^0}\rangle$.
Using the subalgebra $\langle\p_t+C_3w^0\p_{w^0}\rangle$ of this algebra with an arbitrary constant~$C_3$,
we construct the ansatz $w^0=e^{C_3t}\varphi(z)$ with $z=y/\sqrt6$,
which reduces~\eqref{eq:BurgersSystemWithVx=0Uxx=0AndHopf-ColeTrans3} to the Lam\'e equation
\begin{gather}\label{eq:BurgersSystemLameEq}
\varphi_{zz}=6\big(C_3+\wp(z;0,C_1)\big)\varphi.
\end{gather}
As a result, we get the following solution family of the (1+2)-dimensional Burgers equation~\eqref{(1+2)DBurgersEq}
(resp.\ of the Burgers system~\eqref{eq:BurgersSystem} if additionally $v=0$):
\[
u=\wp\left(y/\sqrt6+C_2;0,C_1\right)x+e^{C_3t}\varphi(y/\sqrt6),
\]
where $C_1$, $C_2$ and~$C_3$ are arbitrary constants
and $\varphi$ is the general solution of Lam\'e equation~\eqref{eq:BurgersSystemLameEq}.

In the simplest case $C_1=0$, the above value of~$w^1$ is equivalent, up to translations in~$y$, to the function $w^1=6y^{-2}$.
Then the essential invariance algebra of the equation~\eqref{eq:BurgersSystemWithVx=0Uxx=0AndHopf-ColeTrans3}
becomes wider~\cite[Section~9.9]{ovsiannikov1982}
and, which is more important, all solutions of this equation
can be expressed in terms of solutions of the linear heat equation
using the Darboux transformation,
\[
w^0=\mathrm{DT}[y,y^3+6t]\theta=\theta_{yy}-\frac3y\theta_y+\frac3{y^2}\theta.
\]
Here $\theta=\theta(t,y)$ is an arbitrary solution of the linear heat equation $\theta_t=\theta_{yy}$,
and, given solutions~$f^1$, \dots, $f^k$ of a (1+1)-dimensional linear evolution equation to be transformed with $\mathrm W(f^1,\dots,f^k)\ne0$,
the action of the corresponding Darboux operator on a solution~$f$ is defined by
\[
  \mathrm{DT}[f^1,\dots,f^k]f=\frac{\mathrm W(f^1,\dots,f^k,f)}{\mathrm W(f^1,\dots,f^k)},
\]
where $\mathrm W$'s denote the Wronskians of the indicated tuples of functions with respect to the `space' variable~\cite{popo2008a}.
We construct the solution family of the (1+2)-dimensional Burgers equation~\eqref{(1+2)DBurgersEq}
(again, resp.\ of the Burgers system~\eqref{eq:BurgersSystem} if additionally $v=0$)
\begin{gather}\label{eq:BurgersSystemSolutionSetIn1SolutionOfLHEq}
u=6\frac x{y^2}+\theta_{yy}-\frac3y\theta_y+\frac3{y^2}\theta,
\end{gather}
where $\theta=\theta(t,y)$ is an arbitrary solution of the linear heat equation $\theta_t=\theta_{yy}$.

Finally, setting $w^1$ to be the similarity solution of the nonlinear heat equation with purely quadratic nonlinearity
that is constructed in~\cite{bara2002d} (see also \cite[Section~5.1.1.1.1]{Polyanin&Zaitsev2012})
and finding similarity solutions of the corresponding equation~\eqref{eq:BurgersSystemWithVx=0Uxx=0AndHopf-ColeTrans3},
we obtain the solution
\begin{gather*}
u=12(4\pm\sqrt6)\frac{y^2+(18\pm8\sqrt6)t}{(y^2+10\lambda_\pm t)^2}x
+|t|^{\nu+3/2}\exp\left(-\frac{y^2}{4t}\right)\frac{C_1y+C_2|t|^{1/2}}{(y^2+10\lambda_\pm t)^2}
\\ \phantom{u=}\times
\mathop{\rm HeunC}\left(\frac52\lambda_\pm,-\frac12,-5,\frac58\lambda_\pm(4\nu+1),-\frac52\lambda_\pm\nu-\frac{59}8\mp\frac{29}8\sqrt6,\frac{-y^2}{10\lambda_\pm t}\right),
\end{gather*}
where $\lambda_\pm=3\pm\sqrt6$,
$\mathop{\rm HeunC}(\alpha,\beta,\gamma,\delta,\eta,z)$
is the confluent Heun function,
which is the solution of the following Cauchy problem for the confluent Heun equation with respect to $Y=Y(z)$:
\begin{gather*}
z(z-1)Y_{zz}+\big(\alpha z(z-1)+(\beta+1)(z-1)+(\gamma+1)z\big)Y_z\\
\qquad{}+\frac12\big(\alpha(\beta+1)(z-1)+\alpha(\gamma+1)z+2\delta z+(\beta+1)(\gamma+1)+2\eta-1\big)Y=0,\\
Y(0)=1,\quad Y_z(0)=\frac12\left(\frac{2\eta-1}{\beta+1}+\gamma+1-\alpha\right).
\end{gather*}

\section{Conclusion}\label{sec:Conclusion}

Carrying out symmetry analysis of the two-dimensional Burgers system~\eqref{eq:BurgersSystem} in the present paper,
we have paid a special attention to optimizing all computation.
This was an important ingredient for essentially enhancing and generalizing results on the system~\eqref{eq:BurgersSystem}
that exist in the literature.

Thus, in the course of looking for discrete symmetries of the system~\eqref{eq:BurgersSystem}
and of various related Lie and non-Lie reduced systems with the automorphism-based version of the algebraic method,
we have used the Levi--Malcev theorem and the well-known description of automorphisms of the Lie algebra ${\rm sl}(2,\mathbb R)$
and factored out, from the very beginning, inner automorphisms
of the corresponding maximal Lie invariance algebras that are related to their Levi factors.
We plan to develop this technique and apply it to study discrete symmetries
of other systems of differential equations with finite-dimensional Lie invariance algebras
possessing nonzero Levi factors.

The optimization of the first step of the Lie reduction procedure
is standard for the Burgers system~\eqref{eq:BurgersSystem}
and is to construct optimal lists of one- and two-dimensional subalgebras
of the maximal Lie invariance algebra~$\mathfrak g$ of the system~\eqref{eq:BurgersSystem},
which is the so-called reduced (i.e., centerless) special Galilei algebra with space dimension two.
Although subalgebras of this and other Galilei algebras had been classified, e.g.,
in \cite{bara1989a,bara1995a,FushchychBarannykBarannyk},
this step was not properly implemented in the previous papers
on symmetry analysis of the system~\eqref{eq:BurgersSystem}.
We re-classified one- and two-dimensional subalgebras of~$\mathfrak g$, 
additionally taking into account the external automorphisms of~$\mathfrak g$ 
induced by discrete symmetries of~\eqref{eq:BurgersSystem}. 
This allowed us to find out weaknesses of subalgebra lists for~$\mathfrak g$ that exist in the literature. 
We also re-arranged the constructed subalgebra list for it to be more convenient for the further use.

Due to the application of the original technique of selecting ansatzes for Lie reduction~\cite{popo1995b},
reduced systems are of simple form.
In particular, we have constructed eight reduced systems of (1+1)-dimensional PDEs,
which are at most parameterized by constants
and whose set can be partitioned into two subsets constituted by five and three systems, respectively.
Within each of the subsets, Lie symmetry properties of reduced systems can be studied in a uniform way.
We have shown that all Lie symmetries of reduced systems~1.1--1.7
are induced by Lie symmetries of the original system~\eqref{eq:BurgersSystem},
and reduced system~1.8 is linearized by a Hopf--Cole-type transformation.
This have justified that further Lie reductions of reduced systems~1.1--1.8 are not needed.
Moreover, due to the linearizability of reduced system~1.8,
to study Lie reductions of~\eqref{eq:BurgersSystem} to systems of ODEs,
it suffices to use merely two-dimensional subalgebras of~$\mathfrak g$
that contain no vector fields being $G$-equivalent to the shift vector field~$P^y$.
Only six elements of the obtained optimal list of two-dimensional subalgebras of~$\mathfrak g$
satisfy this property.

Although the Burgers system~\eqref{eq:BurgersSystem} is not linearizable as a whole,
it possesses wide linearizable subsets of solutions.
Besides the well-known solution family
that is singled out by the $G$-invariant differential constraint~\eqref{eq:BurgersSystemLinearizationConstraint}
and parameterized, via the substitution~\eqref{eq:2DHopfColeTrans},
by an arbitrary (nonvanishing) solution of the (1+2)-dimensional linear heat equation~\eqref{eq:(1+2)DLinHeatEq},
we have constructed three families of solutions expressed
in terms of one or two arbitrary (nonvanishing) solutions of the (1+1)-dimensional linear heat equation.
These are the families~\eqref{eq:BurgersSystem2ndSolutionSetIn2SolutionsOfLHEq}
and~\eqref{eq:BurgersSystemSolutionSetIn1SolutionOfLHEq} and
the family defined by ansatz~1.8 jointly with the substitution~\eqref{eq:Hopf--Cole-likeTransForReducedSystem1.8},
and they can additionally be extended by point symmetry transformations of the system~\eqref{eq:BurgersSystem}.
Two more families of exact solutions have been constructed
in Section~\ref{sec:BurgersSystemMoreSolutionsWithMethodOfDiffConstraints}
via the reduction of the system~\eqref{eq:BurgersSystem} to the linear heat equations with potentials.
In the course of looking for exact solutions of the system~\eqref{eq:BurgersSystem},
we have permanently been controlling
whether these solutions do not belong to the above linearizable solution subsets of~\eqref{eq:BurgersSystem}.
In contrast to previous papers on the system~\eqref{eq:BurgersSystem},
we were able to construct new families of exact solutions of~\eqref{eq:BurgersSystem} beyond these subsets.

The most efficient tool for finding exact solutions of the Burgers system~\eqref{eq:BurgersSystem}
is the preliminary reduction of this system to single (1+2)-dimensional PDEs via
posing differential constraints and, in some cases, subsequently introducing potentials.
This is the way of deriving the equation~\eqref{eq:BurgersSystemEqLinearizableTo(1+2)DLinHeatEq},
which is associated with the differential constraint~\eqref{eq:BurgersSystemLinearizationConstraint}
and is linearized by a point transformation
to the (1+2)-dimensional linear heat equation~\eqref{eq:(1+2)DLinHeatEq}.
The reduction to the equation~\eqref{eq:BurgersSystemSingle(1+2)DPDE}
under the differential constraint~\eqref{eq:BurgersSystemDiffConstraintForSingle(1+2)DPDE}
was not known before.
Although the equation~\eqref{eq:BurgersSystemSingle(1+2)DPDE} looks artificial,
its complete point symmetry group is unexpectedly wide.
Moreover, the reduction to this equation allows us to construct
the linearizable solution subset~\eqref{eq:BurgersSystem2ndSolutionSetIn2SolutionsOfLHEq}
as well as the family of solutions of~\eqref{eq:BurgersSystem}
expressed in terms of solutions of the simple complex Hamilton--Jacobi equation $f_t+\frac i2(f_z)^2=0$
and thus parameterized by an arbitrary holomorphic function
(see Section~\ref{sec:SolutionsRelatedToComplexHamiltonJacobiEq}).
The reduction of the system~\eqref{eq:BurgersSystem}
to the (1+2)-dimensional Burgers equation~\eqref{(1+2)DBurgersEq} is obvious.
Up to $G$-equivalence, it is related to the differential constraint $v_x=0$,
which can be interpreted as a constraint for partially invariant solutions.

It is an interesting problem to find new reductions of the Burgers system~\eqref{eq:BurgersSystem}
to single (1+2)-dimensional PDEs,
including such reductions obtainable within the framework of partial invariance~\cite{ovsiannikov1982}.

\section*{Acknowledgments}

ROP and CS would like to express their gratitude for the reciprocal hospitality
shown by their two institutions.
The research of ROP was supported by the Austrian Science Fund (FWF),
project P25064.


\begin{thebibliography}{99}\itemsep=0ex
\footnotesize

\bibitem{abdu2016a}
Abdulwanhhab M.A.,
Exact solutions and conservation laws of system of two-dimensional viscous Burgers equations,
{\it Commun Nonlinear Sci Numer Similat} {\bf 39} (2016), 283--299.

\bibitem{abra2008a}
Abraham-Shrauner B. and Govinder K.S.,
Master partial differential equations for a type II hidden symmetry,
{\it J.~Math. Anal. Appl.} {\bf 343} (2008), 525--530.

\bibitem{ames1965}
Ames W.F.,
{\it Nonlinear partial differential equations in engineering},
Academic Press, New York, 1965.

\bibitem{bara1989a}
Barannik L.F. and Fushchich W.I.,
Continuous subgroups of the generalized Schr\"odinger groups,
{\it J.~Math. Phys.} {\bf 30}  (1989), 280--290.

\bibitem{bara1995a}
Barannyk L.,
On the classification of subalgebras of the Galilei algebras,
{\it J.~Nonlinear Math. Phys.} {\bf 2} (1995), 263--268.

\bibitem{bara2002d}
Barannyk T.,
Symmetry and exact solutions for systems of nonlinear reaction--diffusion equations,
in {\it Proceedings of Fourth International Conference ``Symmetry in Nonlinear Mathematical Physics'' (9--15 July, 2001, Kyiv)},
{\it Proceedings of Institute of Mathematics} {\bf 43} (2002), Part~1, 184--193.

\bibitem{bihl2015a}
Bihlo A., Dos Santos Cardoso-Bihlo E.M. and Popovych R.O.,
Algebraic method for finding equivalence groups,
{\it J. Phys.: Conf. Ser.} {\bf 621} (2015) 012001, 17 pp., arXiv:1503.06487.

\bibitem{bihl2017a}
Bihlo A. and Popovych R.O.,
Group classification of linear evolution equations,
{\it J. Math. Anal. Appl.} {\bf 448} (2017), 982--1005, arXiv:1605.09251.

\bibitem{BlumanBook1989}
Bluman G.W. and Kumei S.,
{\it Symmetries and differential equations},
Springer-Verlag, New York, 1989.

\bibitem{boch1999a}
Bocharov A.V., Chetverikov V.N., Duzhin S.V., Khor'kova N.G., Krasil'shchik I.S., Samokhin A.V., Torkhov~Y.N., Verbovetsky~A.M. and Vinogradov~A.M.,
{\it Symmetries and conservation laws for differential equations of mathematical physics},
American Mathematical Society, Providence, 1999.

\bibitem{broa2015c}
Broadbridge P., 
Classical and quantum Burgers fluids: a challenge for group analysis, 
{\it Symmetry} {\bf 7} (2015), 1803--1815. 

\bibitem{carm2000a}
Carminati J. and Vu K.,
Symbolic computation and differential equations: Lie symmetries,
{\it J.~Symbolic Comput.} {\bf 29} (2000), 95--116.

\bibitem{cole1951}
Cole J.D.,
On a quasi-linear parabolic equation occuring in aerodynamics,
{\it Quart. Appl. Math.} {\bf 9} (1951), 225--236.

\bibitem{deme2008a}
Demetriou E., Ivanova N.M. and Sophocleous C.,
Group analysis of $(2+1)$- and $(3+1)$-dimensional diffusion-convection equations,
{\it J. Math. Anal. Appl.} {\bf 348} (2008), 55--65.

\bibitem{card12a}
Dos Santos Cardoso-Bihlo E.M. and Popovych R.O.,
Complete point symmetry group of the barotropic vorticity equation on a rotating sphere,
{\it J.~Engrg. Math.} {\bf 82} (2013), 31--38, arXiv:1206.6919.

\bibitem{edwa1995a}
Edwards M.P. and Broadbridge P.,
Exceptional symmetry reductions of Burgers' equation in two and three spatial dimensions,
{\it Z.~Angew. Math. Phys.} {\bf 46} (1995), 595--622.

\bibitem{elsa2014a}
El-Sayed M.F., Moatimid G.M., Moussa M.H.M., El-Shiekh R.M. and El-Satar A.A.,
Symmetry group analysis and similarity solutions for the (2+1)-dimensional coupled Burger's system,
{\it Mathematical Methods in the Applied Sciences} {\bf 37} (2014), 1113--1120.

\bibitem{folt1999a}
Foltinek K.,
Conservation laws of evolution equations: generic non-existence,
{\it J.~Math. Anal. Appl.} {\bf 235} (1999), 356--379.

\bibitem{fors1906}
Forsyth~A.R.,
{\it Theory of differential equations}, Vol.~6,
Cambridge University Press, Cambridge, 1906.

\bibitem{FushchychBarannykBarannyk}
Fushchich V.I., Barannik L.F. and Barannik A.F.,
{\it Subgroup analysis of Galilei and Poincare groups and the reduction of nonlinear equations},
Naukova Dumka, Kyiv, 1991. 301 pp. (Russian)
%Fushchych W.I., Barannyk L.F., Barannyk A.F.,
%Symmetry analysis and exact solutions of nonlinear wave equations (draft)

\bibitem{fush1994a}
Fushchych W.I. and Popovych R.O.,
Symmetry reduction and exact solutions of the Navier--Stokes equations.~I,
{\it J.~Nonlinear Math. Phys.} {\bf 1} (1994), 75--113, arXiv:math-ph/0207016.

\bibitem{fush1994b}
Fushchych W.I. and Popovych R.O.,
Symmetry reduction and exact solutions of the Navier--Stokes equations.~II,
{\it J.~Nonlinear Math. Phys.} {\bf 1} (1994), 158--188, arXiv:math-ph/0207016.

\bibitem{hlav1983a}
Hlavat{\'y} L., Steinberg S. and Wolf K.B.,
Linear and nonlinear differential equations as invariants on coset bundles,
{\it Nonlinear phenomena (Oaxtepec, 1982)}, {\it Lecture Notes in Phys.} {\bf 189}, Springer, Berlin--New York, 1983,
pp. 439--451.

\bibitem{hopf1950}
Hopf E.,
The partial differential equation $u_t+uu_x=\mu u_{xx}$,
{\it Comm. Pure. Appl. Math.} {\bf 3} (1950), 201--230.

\bibitem{hydo00b}
Hydon P.E.,
How to construct the discrete symmetries of partial differential equations,
{\it Eur. J. Appl. Math.} {\bf 11} (2000), 515--527.

\bibitem{Ibragimov1985}
Ibragimov~N.H.,
{\it Transformation groups applied to mathematical physics},
D. Reidel Publishing Co., Dordrecht, 1985.

\bibitem{igon2002b}
Igonin S.,
Conservation laws for multidimensional systems and related linear algebra problems,
{\it J.~Phys.~A: Math. Gen.} {\bf 35} (2002), 10607--10617.

\bibitem{king1998a}
Kingston J.G. and Sophocleous C.,
On form-preserving point transformations of partial differential equations,
{\it J.~Phys.~A: Math. Gen.} {\bf 31} (1998), 1597--1619.

\bibitem{kapi1978a}
Kapitanskii L.V.,
Group analysis of the Navier--Stokes and Euler equations in the presence of rotation symmetry and new exact solutions to these equations,
{\it Dokl. Akad. Nauk SSSR} {\bf 243} (1978), 901--904.

\bibitem{kuru2016a}
Kurujyibwami C., Basarab-Horwath P. and Popovych R.O.,
Algebraic method for group classification of (1+1)-dimensional linear Schr\"odinger equations,
2016, arXiv:1607.04118, 30 pp.

\bibitem{mart79a}
Mart\'{i}nez-Alonso L.,
On the Noether map,
{\it Lett. Math. Phys.} {\bf 3} (1979), 419--424.

\bibitem{olve1993b}
Olver P.J.,
{\it Applications of Lie groups to differential equations}, second edition,
{\it Graduate Texts in Mathematics} {\bf 107}, Springer--Verlag, New York, 1993.

\bibitem{ovsiannikov1982}
Ovsiannikov L.V.,
{\it Group analysis of differential equations},
Academic Press, New York, 1982.

\bibitem{pate1977a}
Patera J. and Winternitz P.
Subalgebras of real three and four-dimensional Lie algebras,
{\it J.~Math. Phys.} {\bf 18} (1977), 1449--1455.

\bibitem{poch2013a}
Pocheketa O.A. and Popovych R.O.,
Reduction operators of Burgers equation,
{\it J.~Math. Anal. Appl.} {\bf 398} (2013), 270--277, arXiv:1208.0232.

\bibitem{Polyanin&Zaitsev2012}
Polyanin A.D. and Zaitsev V.F.,
{\it Handbook of nonlinear partial differential equations}, second edition,
Chapman \& Hall/CRC, Boca Raton, FL, 2012.

\bibitem{popo1995b}
Popovych R.O.,
On Lie reduction of the Navier--Stokes equations,
{\it J.~Nonlinear Math. Phys.} {\bf 2} (1995), 301--311.

\bibitem{popo2005b}
Popovych R.O. and Ivanova N.M.,
Hierarchy of conservation laws of diffusion-convection equations,
{\it J.~Math. Phys.} {\bf 46} (2005), 043502, 22 pp., arXiv:math-ph/0407008.

\bibitem{popo2008a}
Popovych R.O., Kunzinger M. and Ivanova N.M.,
Conservation laws and potential symmetries of linear parabolic equations,
{\it Acta Appl. Math.} {\bf 100} (2008), 113--185, arXiv:0706.0443.

\bibitem{raja2008a}
Rajaee L., Eshraghi H. and Popovych R.O.,
Multi-dimensional quasi-simple waves in weakly dissipative flows,
{\it Physica D} {\bf 237} (2008), 405--419.

\bibitem{sach1987a}
Sachdev P.L.,
{\it Nonlinear diffusive waves},
Cambridge University Press, New York, 1987.

\bibitem{sale1987a}
Salerno M.,
On the phase manifold geometry of the two-dimensional Burgers equation,
{\it Phys. Lett. A} {\bf 121} (1987), 15--18.

\bibitem{tami1991a}
Tamizhmani K.M. and Punithavathi P.,
Similarity reductions and Painlev\'e property of the coupled higher-dimensional Burgers' equation.
{\it Internat. J. Non-Linear Mech.} {\bf 26} (1991), 427--438.

\bibitem{vino1984a}
Vinogradov A.M.,
Local symmetries and conservation laws,
{\it Acta Appl. Math.} {\bf 2} (1984), 21--78.

\bibitem{vu2007a}
Vu K.T., Butcher J. and Carminati J.,
Similarity solutions of partial differential equations using DESOLV,
{\it Comput. Phys. Comm.} {\bf 176} (2007), 682--693.

\end{thebibliography}
\end{document}